\documentclass[kpfonts]{patmorin}
\usepackage{pat}
\usepackage{paralist}
\usepackage{dsfont}  
\usepackage[utf8]{inputenc}

\usepackage{graphicx}
\usepackage[noend]{algorithmic}

\usepackage{xcolor}
\definecolor{light-gray}{gray}{0.95}

\usepackage[normalem]{ulem}
\usepackage{cancel}
\usepackage{enumitem}

\usepackage{comment}

\newcommand{\Oh}{\mathcal{O}}

\let\le\leqslant
\let\ge\geqslant
\let\leq\leqslant
\let\geq\geqslant

\let\preceq\preccurlyeq

\newcommand{\itemref}[1]{(\ref{#1})}

\title{\MakeUppercase{Adjacency Labelling for Planar Graphs (and Beyond)}}
\author{
  Vida Dujmović,%
    \thanks{School of Computer Science and Electrical Engineering, University of Ottawa, Canada. This research was partially supported by NSERC.}\quad
  Louis Esperet,%
    \thanks{Laboratoire G-SCOP (CNRS, Univ.\ Grenoble Alpes), Grenoble, France. Partially supported by the French ANR Projects ANR-16-CE40-0009-01 (GATO) and ANR-18-CE40-0032 (GrR).}\quad
  Cyril Gavoille,%
    \thanks{LaBRI, University of Bordeaux, France. This research was partially supported by the French ANR projects ANR-16-CE40-0023 (DESCARTES) and ANR-17-CE40-0015 (DISTANCIA).}\quad
  Gwenaël Joret,%
     \thanks{Département d'Informatique, Université Libre de Bruxelles, Brussels, Belgium. Research supported by an ARC grant from the Wallonia-Brussels Federation of Belgium and by a grant from the National Fund for Scientific Research (FNRS).}\quad
  Piotr Micek,%
    \thanks{Theoretical Computer Science Department, Faculty of Mathematics and Computer Science, Jagiellonian University, Krak\'{o}w, Poland. This research was partially supported by the Polish National Science Center grant (BEETHOVEN; UMO-2018/31/G/ST1/03718).}\newline
  and Pat Morin%
    \thanks{School of Computer Science, Carleton University, Canada. This research was partially supported by NSERC.}
}
\date{}

\begin{document}
\begin{titlepage}
\maketitle

\begin{abstract}
  We show that there exists an adjacency labelling scheme for planar graphs where each vertex of an $n$-vertex planar graph $G$ is assigned a $(1+o(1))\log_2 n$-bit label and the labels of two vertices $u$ and $v$ are sufficient to determine if $uv$ is an edge of $G$.  This is optimal up to the lower order term and is the first such asymptotically optimal result.  An alternative, but equivalent, interpretation of this result is that, for every positive integer $n$, there exists a graph $U_n$ with $n^{1+o(1)}$ vertices such that every $n$-vertex planar graph is an induced subgraph of $U_n$.  These results generalize to a number of other graph classes, including bounded genus graphs, apex-minor-free graphs,  bounded-degree graphs from minor closed families, and $k$-planar graphs.
\end{abstract}
\end{titlepage}
\pagenumbering{roman}
\tableofcontents

\newpage

\setcounter{page}{0}
\pagenumbering{arabic}
\section{Introduction}


A family $\mathcal{G}$ of graphs has an \emph{$f(n)$-bit adjacency labelling scheme} if there exists a function $A:(\{0,1\}^*)^2\to \{0,1\}$ such that for every $n$-vertex graph $G\in \mathcal{G}$ there exists $\ell:V(G)\to\{0,1\}^*$ such that $|\ell(v)|\le f(n)$ for each vertex $v$ of $G$ and such that, for every two vertices $v,w$ of $G$,
\[  A(\ell(v),\ell(w)) =
      \begin{cases}
        0 & \text{if $vw\not\in E(G)$;} \\
        1 & \text{if $vw\in E(G)$.}
      \end{cases}
\]

Let $\log x:=\log_2 x$ denote the binary logarithm of $x$.
In this paper we prove the following result:
\begin{thm}\thmlabel{main}
  The family of planar graphs has a $(1+o(1))\log n$-bit adjacency labelling scheme.
\end{thm}


\thmref{main} is optimal up to the lower order $o(\log n)$ term,
which is $\Oh\left(\sqrt{\log n\log\log n}\right)$ in our proof.
An alternative, but equivalent, interpretation of \thmref{main} is that, for every integer $n\ge 1$, there exists a graph $U_n$ with $n^{1+o(1)}$  vertices such that every $n$-vertex planar graph is isomorphic to some vertex-induced subgraph of $U_n$.\footnote{There is a small technicality that the equivalence between adjacency labelling schemes and universal graphs requires that $\ell:V(G)\to\{0,1\}^*$ be injective.  The labelling schemes we discuss satisfy this requirement.  For more details about the connection between labelling schemes and universal graphs, the reader is directed to Spinrad's monograph \cite[Section~2.1]{spinrad:efficient}.}

Note that the proof of \thmref{main} is constructive: it gives an algorithm producing the labels in $\Oh(n\log n)$ time.

\subsection{Previous Work}

The current paper is the latest in a series of results dating back to Kannan, Naor, and Rudich \cite{kannan.naor.ea:implicit0,kannan.naor.ea:implicit} and Muller \cite{muller:local} who defined adjacency labelling schemes\footnote{There are some small technical differences between the  definitions in \cite{kannan.naor.ea:implicit} and \cite{muller:local} that have to do with the complexity of computing $\ell(\cdot)$ as a function of $G$ and of computing
$A(\cdot,\cdot)$ as a function of its two arguments.} and described $\Oh(\log n)$-bit adjacency labelling schemes for several classes of graphs, including planar graphs.  Since this initial work, adjacency labelling schemes and, more generally, informative labelling schemes have remained a very active area of research \cite{adjiashvili.rotbart:labeling,alstrup.kaplan.ea:adjacency,abrahamsen.alstrup.ea:near-optimal,alstrup.dahlgaard.ea:sublinear,alstrup.gortz.ea:distance,alstrup.gavoille.ea:simpler,alstrup.rauhe:improved,Alon17}.

Here we review results most relevant to the current work, namely results on planar graphs and their supporting results on trees and bounded-treewidth graphs.  First, a superficial review: Planar graphs have been shown to have $(c+o(1))\log n$-bit adjacency labelling schemes for successive values of $c=6,4,3,2,\tfrac{4}{3}$ and finally \thmref{main} gives the optimal\footnote{It is easy to see that, in any adjacency labelling scheme for any $n$-vertex graph $G$ in which no two vertices have the same neighbourhood, all labels must be distinct, so some label must have length at least $\lceil\log n\rceil$.} result $c=1$.  We now give details of these results.

Muller's scheme for planar graphs \cite{muller:local} is based on the fact that planar graphs are 5-degenerate.  This scheme orients the edges of the graph so that each vertex has 5 outgoing edges, assigns each vertex $v$ an arbitrary $\lceil\log n\rceil$-bit identifier, and assigns a label to $v$ consisting of $v$'s identifier and the identifiers of the targets of $v$'s outgoing edges.  In this way, each vertex $v$ is assigned a label of length at most $6\lceil\log n\rceil$.  Kannan, Naor, and Rudich \cite{kannan.naor.ea:implicit} use a similar approach that makes use of the fact that planar graphs have arboricity 3 (so their edges can be partitioned into three forests \cite{nash-williams:edge-disjoint}) to devise an adjacency labelling scheme for planar graphs whose labels have length at most $4\lceil\log n\rceil$.

A number of $(1+o(1))\log n$-bit adjacency labelling schemes for forests have been devised \cite{chung:universal, alstrup.rauhe:improved,alstrup.dahlgaard.ea:optimal}, culminating with a recent $(\log n + \Oh(1))$-bit adjacency labelling scheme \cite{alstrup.dahlgaard.ea:optimal} for forests.  Combined with the fact that planar graphs have arboricity 3, these schemes imply $(3+o(1))\log n$-bit adjacency labelling schemes for planar graphs.

A further improvement, also based on the idea of partitioning the edges of a planar graph into simpler graphs was obtained by Gavoille and Labourel \cite{gavoille.labourel:shorter}.  Generalizing the results for forests, they describe a $(1+o(1))\log n$-bit adjacency labelling scheme for $n$-vertex graphs of bounded treewidth. As is well known, the edges of a planar graph can be partitioned into two sets, each of which induces a bounded treewidth graph \cite{goncalves.gabow.ea:edge}.
This results in a $(2+o(1))\log n$-bit adjacency labelling scheme for planar graphs.

Very recently, Bonamy, Gavoille, and Pilipczuk \cite{bonamy.gavoille.ea:shorter} described a $(4/3+o(1))\log n$-bit adjacency labelling scheme for planar graphs based on a recent \emph{graph product structure theorem} of Dujmović \etal\ \cite{dujmovic.joret.ea:planar}.  This product structure theorem states that any planar graph is a subgraph of a strong product $H\boxtimes P$ where $H$ is a bounded-treewidth graph and $P$ is a path. See \figref{product}. It is helpful to think of $H\boxtimes P$ as a graph whose vertices can be partitioned into $h:=|V(P)|$ \emph{rows} $H_1,\dots,H_{h}$, each of which induces a copy of $H$ and with vertical and diagonal edges joining corresponding and adjacent vertices between consecutive rows.

\begin{figure}[htbp]
  \begin{center}
    \includegraphics{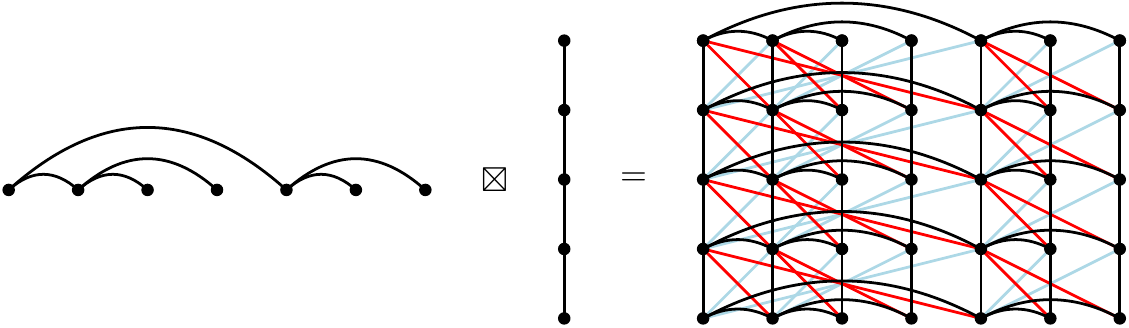}
  \end{center}
  \caption{The strong product $H\boxtimes P$ of a tree $H$ and a path $P$.}
  \figlabel{product}
\end{figure}

The product structure theorem quickly leads to a $(1+o(1))\log(mh)$-bit labelling scheme where $m:=|V(H)|$ and $h:=|V(P)|$ by using a $(1+o(1))\log m$-bit labelling scheme for $H$ (a bounded treewidth graph), a $\ceil{\log{h}}$-bit labelling scheme for $P$ (a path), and a constant number of bits to locally encode the subgraph of $H\boxtimes P$ (of constant arboricity).  However, for an $n$-vertex graph $G$ that is a subgraph of $H\boxtimes P$ in the worst case $m$ and $h$ are each $\Omega(n)$, so this offers no immediate improvement over the existing $(2+o(1))\log n$-bit scheme.

Bonamy, Gavoille, and Pilipczuk improve upon this by cutting $P$ (and hence $G$) into subpaths of length $n^{1/3}$ in such a way that this corresponds to removing $\Oh(n^{2/3})$ vertices of $G$ that have a neighbourhood of size $\Oh(n^{2/3})$. The resulting (cut) graph is a subgraph of $H'\boxtimes P'$ where $H'$ has bounded treewidth, $|H'|\le n$, and $P'$ is a path of length $n^{1/3}$ so it has a labelling scheme in which each vertex has a label of length $(1+o(1))\log (|H'|\cdot|P'|) \le (4/3+o(1))\log n$.  A slight modification of this scheme allows for the $\Oh(n^{2/3})$ \emph{boundary} vertices adjacent to the cuts to have shorter labels, of length only $(2/3+o(1))\log n$.  The cut vertices and the boundary vertices induce a bounded-treewidth graph of size $\Oh(n^{2/3})$.  The vertices in this graph receive secondary labels of length $(2/3+o(1))\log n$.  In this way, every vertex receives a label of length at most $(4/3 + o(1))\log n$.

\subsection{New Results}

The adjacency labelling scheme described in the current paper is also based on the product structure theorem for planar graphs, but it avoids cutting the path $P$, and thus avoids boundary vertices that take part in two different labelling schemes.  Instead, it uses a weighted labelling scheme on the rows $H_1,\dots,H_h$ of $H\boxtimes P$ in which vertices that belong to $H_i$ receive a label of length $(1+o(1))\log n-\log W_i$ where $W_i$ is related to the number of vertices of $G$ contained in $H_i$ and $H_{i-1}$.  The vertices of $G$ in row $i$ participate in a secondary labelling scheme for the subgraph of $G$ contained in $H_i$ and $H_{i-1}$ and the labels in this scheme have length $\log W_i + o(\log n)$. Thus every vertex receives two labels, one of length $(1+o(1))\log n-\log W_i$ and another of length $\log W_i + o(\log n)$ for a total label length of $(1+o(1))\log n$.

The key new technique that allows all of this to work is that the labelling schemes of the rows $H_1,\dots,H_h$ are not independent.  All of these labelling schemes are based on a single balanced binary search tree $T$ that undergoes insertions and deletions resulting in a sequence of related binary search trees $T_1,\dots,T_h$ where each $T_i$ represents all vertices of $G$ in $H_{i}$ and $H_{i-1}$ and the label assigned to a vertex of $H_i$ is essentially based on a path from the root of $T_i$ to some vertex of $T_i$.  By carefully maintaining the binary search tree $T$, the trees $T_{i-1}$ and $T_{i}$ are similar enough so that the label for $v$ in $H_i$ can be obtained, with $o(\log n)$ additional bits from the label for $v$ in $H_{i-1}$.

The product structure theorem has been generalized to a number of additional graph families including bounded-genus graphs, apex-minor free graphs, bounded-degree graphs from minor-closed families, $k$-planar graphs, powers of bounded-degree bounded genus graphs, and $k$-nearest neighbour graphs of points in $\R^2$ \cite{dujmovic.joret.ea:planar,dujmovic.morin.ea:structure}. As a side-effect of designing a labelling scheme to work directly on subgraphs of a strong product $H\boxtimes P$, where $H$ has bounded treewidth and $P$ is a path,
we obtain $(1+o(1))\log n$-bit labelling schemes for all of these graph families.  All of these results are optimal up to the lower order term.

A graph is \emph{apex} if it has a vertex whose removal leaves a planar graph.
A graph is \emph{$k$-planar} if it has a drawing in the plane in which
each edge is involved in at most $k$ crossings. Such graphs provide a natural
generalisation of planar graphs, and have been extensively studied~\cite{kobourov.liotta.ea:annotated}.
The definition of $k$-planar graphs naturally generalises for other surfaces. A graph $G$ is
\emph{$(g,k)$-planar} if it has a drawing in some surface of Euler genus at most
$g$ in which each edge of $G$ is involved in at most $k$ crossings.
Note that already $1$-planar graphs are not minor closed.
The generalization of \thmref{main} provided by known product structure theorems is summarized in the following result:

\begin{thm}\thmlabel{main-all}
  For every fixed integer $t\geq 1$, the family of all graphs $G$ such that $G$ is a subgraph of $H\boxtimes P$ for some graph $H$ of treewidth $t$ and some path $P$ has a $(1+o(1))\log n$-bit adjacency labelling scheme.
  This includes the following graph classes:
  \begin{compactenum}
    \item graphs of bounded genus and, more generally, apex-minor free graphs;
    \item bounded degree graphs that exclude a fixed graph as a minor; and
    \item $k$-planar graphs and, more generally, $(g,k)$-planar graphs.
  \end{compactenum}
\end{thm}

The case of graphs of bounded degree from minor-closed classes (point~2 in \thmref{main-all}) is particularly interesting since, prior to the current work, the best known bound for adjacency labelling schemes in planar graphs of bounded degree was the same as for general planar graphs, i.e., $(4/3+o(1))\log n$. On the other hand, our \thmref{main-all} now gives an asymptotically optimal bound of $(1+o(1))\log n$ for graphs of bounded degree from any proper minor-closed class.

\subsection{Outline}

The remainder of the paper is organized as follows. \Secref{preliminaries} reviews some preliminary definitions and easy results.  \Secref{bulk-trees} describes a new type of balanced binary search tree that has the specific properties needed for our application. \Secref{pxp} solves a special case, where $G$ is an $n$-vertex subgraph of $P_1\boxtimes P_2$ where $P_1$ and $P_2$ are both paths. We include it to highlight the generic idea behind our adjacency labelling scheme. \Secref{hxp} solves the general case in which $G$ is an $n$-vertex subgraph of $H\boxtimes P$ where $H$ has bounded treewidth and $P$ is a path.  \Secref{conclusion} concludes with a discussion of the computational complexity of assigning labels and testing adjacency and presents directions for future work.

\section{Preliminaries}
\seclabel{preliminaries}

All graphs we consider are finite and simple.  The vertex and edge sets of a graph $G$ are denoted by $V(G)$ and $E(G)$, respectively.  The \emph{size} of a graph $G$ is denoted by $|G|:=|V(G)|$.

For any graph $G$ and any vertex $v\in V(G)$, let $N_G(v):=\{w\in V(G): vw\in E(G)\}$ and $N_G[v]:=N_G(v)\cup\{v\}$ denote the open neighbourhood and closed neighbourhood of $v$ in $G$, respectively.

\subsection{Prefix-Free Codes}

For a string $s=s_1,\dots,s_k$, we use $|s|:=k$ to denote the length of $s$.
A string $s_1,\dots,s_k$ is a \emph{prefix} of a string $t_1,\dots,t_\ell$ if $k\le \ell$ and $s_1,\dots,s_k=t_1,\dots,t_k$.  A \emph{prefix-free code} $c:X\to\{0,1\}^*$ is a one-to-one function in which $c(x)$ is not a prefix of $c(y)$ for any two distinct $x,y\in X$.  Let $\N$ denote the set of non-negative integers.  The following is a classic observation\footnote{The binary representation $w$ of a positive integer $i$ has length  $|w|=\lfloor\log i\rfloor+1$ and begins with $1$. The gamma code for $i$ is given by $\gamma(i):=0^{|w|-1}w$. This give a codeword that is decoded by counting the number, $z$, of leading zeros and then treating the next $z+1$ bits as the binary representation of a positive integer.} of Elias from 1975.

\begin{lem}[Elias \cite{elias:universal}]\lemlabel{elias}
  There exists a prefix-free code $\gamma:\N\to\{0,1\}^*$ such that, for each $i\in\N$, $|\gamma(i)|\le 2\lfloor\log(i+1)\rfloor + 1$.
  \end{lem}

  In the remainder of the paper, $\gamma$ (which we call an \emph{Elias encoding}) will be used extensively, without referring systematically to \lemref{elias}.

\subsection{Labelling Schemes Based on Binary Trees}

A \emph{binary tree} $T$ is a rooted tree in which each node except the root is either the \emph{left} or \emph{right} child of its parent and each node has at most one left and at most one right child.  For any node $x$ in $T$, $P_T(x)$ denotes the path from the root of $T$ to $x$.  The \emph{length} of a path $P$ is the number of edges in $P$, i.e., $|P|-1$.  The \emph{depth}, $d_T(x)$ of $x$ is the length of $P_T(x)$.  The \emph{height} of $T$ is $h(T):=\max_{x\in V(T)} d_T(x)$.  A \emph{perfectly balanced} binary tree is any binary tree $T$ of height $h(T)=\lfloor\log|T|\rfloor$, where $|T|:=|V(T)|$ denotes the number of nodes in $T$.

A binary tree is \emph{full} if each non-leaf node has exactly two children. For a binary tree $T$, we let $T^+$ denote the full binary tree obtained by attaching to each node $x$ of $T$ $2-c_x$ leaves where $c_x \in\{0,1,2\}$ is the number of children of $x$.  We call the leaves of $T^+$ the \emph{external nodes} of $T$.  (Note that none of these external nodes are in $T$ and, for any non-empty $T$,  $h(T^+)=h(T)+1$.)

A node $a$ in $T$ is a \emph{$T$-ancestor} of a node $x$ in $T$ if $a\in V(P_T(x))$. If $a$ is a $T$-ancestor of $x$ then $x$ is a \emph{$T$-descendant} of $a$. (Note that a node is a $T$-ancestor and $T$-descendant of itself.)  For a subset of nodes $X\subseteq V(T)$, the \emph{lowest common $T$-ancestor} of $X$ is the maximum-depth node $a\in V(T)$ such that $a$ is a $T$-ancestor of $x$ for each $x\in X$.

Let $P_T(x_r)=x_0,\dots,x_{r}$ be a path from the root $x_0$ of $T$ to some node $x_r$ (possibly $r=0$).  Then the \emph{signature} of $x_r$ in $T$, denoted $\sigma_T(x_r)$ is a binary string $b_1,\dots,b_r$ where $b_i=0$ if and only if $x_{i}$ is the left child of $x_{i-1}$.
Note that the signature of the root of $T$ is the empty string.

A \emph{binary search tree} $T$ is a binary tree whose node set $V(T)$ consists of distinct real numbers and that has the \emph{binary search tree property}:  For each node $x$ in $T$, $z<x$ for each node $z$ in $x$'s left subtree and $z>x$ for each node $z$ in $x$'s right subtree. For any $x\in\R\setminus V(T)$, the \emph{search path} $P_T(x)$ in $T$ is the unique root-to-leaf path $v_0,\dots,v_r$ in $T^+$ such that adding $x$ as a (left or right, as appropriate) child of $v_{r-1}$ in $T$ would result in a binary search tree $T'$ with $V(T')=V(T)\cup\{x\}$.

The following observation allows us to compare values in a binary search tree just given their signatures in the tree.

\begin{obs}\obslabel{lexicographic}
  For any binary search tree $T$ and any nodes $x$, $y$ in $T$, we have $x<y$ if and only if $\sigma_T(x)$ is lexicographically less than $\sigma_T(y)$.
\end{obs}

Let $\R^+$ denote the set of positive real numbers. The following is a folklore result about biased binary search trees, but we sketch a proof here for completeness.

\begin{lem}\lemlabel{biased-bst}
  For any finite $S\subset \R$ and any function $w:S\to\R^+$, there exists a binary search tree $T$ with $V(T)=S$ such that, for each $y\in S$, $d_T(y)\le\log(W/w(y))$, where $W:=\sum_{y\in S} w(y)$.
\end{lem}

\begin{proof}
  The proof is by induction on $|S|$. The base case $|S|=0$ is vacuously true.
  For any $x\in\R$, let $S_{<x}:=\{y\in S: y < x\}$ and $S_{>x}:=\{y\in S: y>x\}$. For $|S|\ge 1$, choose the root of $T$ to be the unique node $y_0\in S$ such that $\sum_{z\in S_{<y_0}} w(z)\le W/2$ and $\sum_{z\in S_{>y_0}}< W/2$. Apply induction on $S_{<y_0}$ and $S_{>y_0}$ to obtain the left and right subtrees of $T$, respectively.

  Then $d_T(y_0)=0=\log 1\le \log (W/w(y_0))$.  For each $y\in S_{<y_0}$,
  \[
    d_T(y) \le 1 + \log\left(\frac{\sum_{z\in S_{<y_0}}w(z)}{w(y)}\right)
            \le 1 + \log \left(\frac{W/2}{w(y)}\right)
            = \log \left(\frac{W}{w(y)}\right) ,
  \]
  and the same argument shows that $d_T(y) < \log (W/w(y))$ for each $y\in S_{>y_0}$.
\end{proof}

The following fact about binary search trees is useful, for example, in the deletion algorithms for several types of balanced binary search trees \cite[Section~6.2.3]{morin:open}, see \figref{sigma-minus-one}:

\begin{figure}
  \begin{center}
    \includegraphics{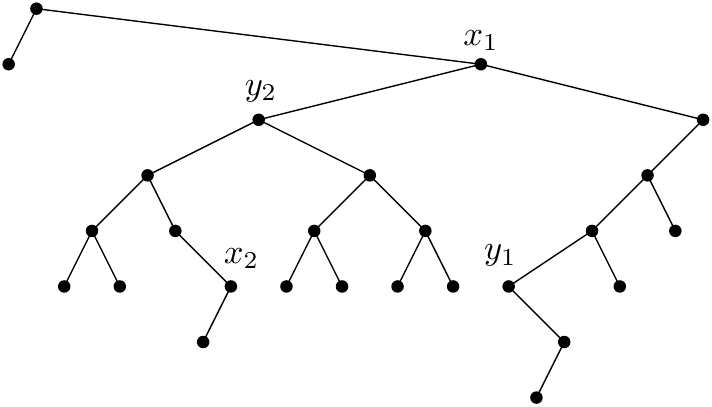}
  \end{center}
  \caption{An illustration of \obsref{predecessor-encoding}: (1)~$\sigma_T(y_1)=11000$ and $\sigma_T(x_1)=1\mbox{\sout{1000}}$ (2)~$\sigma_T(y_2)=10\mbox{\sout{011}}$ and $\sigma_T(x_2)=10011$.}
  \figlabel{sigma-minus-one}
\end{figure}

\begin{obs}\obslabel{predecessor-encoding}
  Let $T$ be a binary search tree and let $x$, $y$ be nodes in $T$ such that $x<y$ and there is no node $z$ in $T$ such that $x<z<y$, i.e., $x$ and $y$ are consecutive in the sorted order of $V(T)$.  Then
  \begin{enumerate}
    \item (if $y$ has no left child) $\sigma_T(x)$ is obtained from $\sigma_T(y)$ by removing all trailing 0's and the last 1; or
    \item (if $y$ has a left child) $\sigma_T(x)$ is obtained from $\sigma_T(y)$ by appending a 0 followed by $d_T(y)-d_T(x)-1$ 1's.
  \end{enumerate}
Therefore, there exists a function $D:(\{0,1\}^*)^2\to\{0,1\}^*$ such that, for every binary search tree $T$ and for every two consecutive nodes $x, y$ in the sorted order of $V(T)$, there exists $\delta_T(y) \in \{0,1\}^*$ with $|\delta_T(y)|=\Oh(\log h(T))$ such that
$D(\sigma_T(y),\delta_T(y))=\sigma_T(x)$.
\end{obs}

The bitstring $\delta_T(y)$ from \obsref{predecessor-encoding} is obtained as follows: It consists of a first bit indicating whether $y$ has a left child in $T$ or not and, in case $y$ does have a left child, an Elias encoding $\gamma(s)$ of the value $s:=d_T(y)-d_T(x)-1$.  More precisely, $\delta_T(y) = 0$ or $\delta_T(y) = 1, \gamma(s)$.

Putting some of the preceding results together we obtain the following useful coding result.

\begin{lem}\lemlabel{row-code}
  There exists a function $A:(\{0,1\}^*)^2\to\{-1,0,1,\perp\}$ such that, for any $h\in\N$, and any $w:\{1,\dots,h\}\to\R^+$ there is a prefix-free code $\alpha:\{1,\dots,h\}\to \{0,1\}^*$ such that
  \begin{compactenum}
    \item for each $i\in\{1,\dots,h\}$, $|\alpha(i)|=\log W -\log w(i) + \Oh(\log\log h)$, where $W:=\sum_{j=1}^h w(j)$;
    \item for any $i,j\in\{1,\dots,h\}$,
    \[   A(\alpha(i),\alpha(j))
    = \begin{cases}
       0 & \text{if $j=i$;}\\
       1 & \text{if $j=i+1$;} \\
       -1 & \text{if $j=i-1$;} \\
       \perp & \text{otherwise.}
      \end{cases}
      \]
    \end{compactenum}
\end{lem}

\begin{proof}
  Define $w':\{1,\dots,h\}\to \R^+$ as $w'(i)=w(i)+W/h$ and let $W':=\sum_{i=1}^h w'(i)=2W$.
  Using \lemref{biased-bst}, construct a biased binary search tree $T$ on $\{1,\dots,h\}$ using $w'$ so that
  \[
    d_T(i)\le\log (2W)-\log(w(i)+W/h) \le \log W-\log w(i)+1
  \]
  and
  \[
  d_T(i)\le\log (2W)-\log(w(i)+W/h) \le \log W-\log (W/h)+1 \le \log h + 1,
  \]
  for each $i\in\{1,\dots,h\}$.  This latter inequality implies that $h(T)\le\log h + 1$.

  The code $\alpha(i)$ for $i$ consists of three parts.  The first part, $\gamma(|\sigma_T(i)|)$, is the Elias encoding of the length of the path $P_T(i)$ from the root to $i$ in $T$. The second part $\sigma_T(i)$ encodes the left/right turns along this path. The third part, $\delta_T(i)$, is defined in \obsref{predecessor-encoding}.
  The length of $\delta_T(i)$ is $\Oh(\log h(T))=\Oh(\log\log h)$. Note that since $\gamma$ is prefix-free and two distinct sequences $\sigma_T(i)$ and $\sigma_T(j)$ of the same length cannot be prefix of one another, the code $\alpha$ is also prefix-free (and thus injective).

  The function $A$ is given by a simple algorithm: Given $\alpha(i)$ and $\alpha(j)$, first observe that the values of $\gamma(\cdot)$, $\sigma_T(\cdot)$, and $\delta_T(\cdot)$ can be extracted: $\gamma(\cdot)$ is first extracted using the fact that Elias encoding is prefix-free, this then gives us the length of $\sigma_T(\cdot)$, and finally $\delta_T(\cdot)$ consists of the remaining bits.
   The function $A$ extracts the values and lexicographically compares $\sigma_T(i)$ and $\sigma_T(j)$.  If $\sigma_T(i)=\sigma_T(j)$, then $A$ outputs $0$.
   Otherwise, assume for now that $\sigma_T(i)$ is lexicographically less than $\sigma_T(j)$ so that, by \obsref{lexicographic}, $i < j$.  Now $A$ computes $D(\sigma_T(j),\delta_T(j))=\sigma_T(j-1)$ as described in \obsref{predecessor-encoding}.
   If $\sigma_T(j-1)=\sigma_T(i)$ then $A$ outputs $1$, otherwise $A$ outputs $\perp$.
  In the case where $\sigma_T(i)$ is lexicographically greater than $\sigma_T(j)$, $A$ proceeds in the same manner, but reversing the roles of $i$ and $j$ and outputting $-1$ in the case where $\sigma_T(i-1)=\sigma_T(j)$.
\end{proof}

\subsection{Chunked Sets}

For non-empty finite sets $X,Y\subset \R$ and an integer $a$, we say that $X$ \emph{$a$-chunks} $Y$ if, for any $a+1$-element subset $S\subseteq Y$, there exists $x\in X$, such that $\min S\le x\le \max S$. Observe that, if $X$ $a$-chunks $Y$, then $|Y\setminus X|\le a(|X|+1)\le 2a|X|$ so $|X\cup Y|\le (2a+1)|X|$.  A sequence $V_1,\dots,V_h$ of non-empty subsets of $\R$ is \emph{$a$-chunking} if $V_y$ $a$-chunks $V_{y+1}$ and $V_{y+1}$ $a$-chunks $V_y$ for each $y\in\{0,\dots,h-1\}$.

\begin{lem}\lemlabel{fractional}
  For any finite sets $S_1,\dots,S_h\subset\R$ and any integer $a\ge 1$, there exist sets $V_1,\dots,V_{h}\subset\R$ such that
  \begin{compactenum}
    \item for each $y\in\{1,\dots,h\}$, $V_y\supseteq S_y$;
    \item $V_1,\dots,V_h$ is $a$-chunking;
    \item $\sum_{y=1}^h |V_y|\le \left((a+1)/a\right)^2\cdot\sum_{y=1}^h |S_y|$.
  \end{compactenum}
\end{lem}

A proof of a much more general version of \lemref{fractional} (with larger constants) is implicit in the iterated search structure of Chazelle and Guibas \cite{chazelle.guibas:fractional1}.   For the sake of completeness, we give a proof of \lemref{fractional} below. In the remainder of the paper, we will use \lemref{fractional} with $a=1$. In that case the third item above becomes $\sum_{y=1}^h |V_y|\le 4 \cdot\sum_{y=1}^h |S_y|$.

 \begin{proof}[Proof of \lemref{fractional}]
  Set $W_0=\emptyset$ and then for each $y:=1,\ldots,h$ repeat the following procedure:  Let $W_{y-1}'$ consist of every $(a+1)$th element of the sequence obtained by sorting $W_{y-1}$, beginning with the $(a+1)$th element of this sequence.  Observe that $W_{y-1}'$ has size $\lfloor|W_{y-1}|/(a+1)\rfloor \le |W_{y-1}|/(a+1)$.  Now set $W_y:= S_y\cup W_{y-1}'$.

  After the final ($y=h$) iteration, it is clear that $W_y\supseteq S_y$ and that $W_y$ $a$-chunks $W_{y-1}$ for each $y\in\{1,\ldots,h\}$.  Furthermore, for each $y\in\{1,\ldots,h\}$,
  \[
      |W_y| = |S_y|+|W_{y-1}'| \le |S_y|+|W_{y-1}|/(a+1)
  \]
  and an easy proof by induction on $y$ shows that
  \[
      |W_y| \le \sum_{i=0}^{y-1} |S_{y-i}|/(a+1)^i \enspace .
  \]
  Then
  \[
      \sum_{y=1}^h |W_y| \le \sum_{y=1}^h\sum_{i=0}^{y-1} |S_{y-i}|/(a+1)^i
      \le \sum_{y=1}^h\sum_{i=0}^\infty |S_{y}|/(a+1)^i
      = \frac{a+1}{a} \sum_{y=1}^h |S_y| \enspace .
  \]
  Next, observe that copying any element from $W_y$ into $W_{y-1}$ preserves the fact that $W_y$ $a$-chunks $W_{y-1}$.  This allows us to use exactly the same procedure on the sequence of sets $W_{h},\ldots,W_{1}$ to obtain an $a$-chunking sequence $V_{1},\ldots,V_h$ such that $V_y\supseteq W_y\supseteq S_y$ for each $y\in\{1,\ldots,h\}$ and
  \[  \sum_{y=1}^h |V_y| \le \frac{a+1}{a}\sum_{y=1}^h |W_y| \le \left(\frac{a+1}{a}\right)^2\cdot\sum_{y=1}^h |S_y| \enspace . \qedhere
  \]
\end{proof}

\subsection{Product Structure Theorems}

The \emph{strong product} $A\boxtimes B$ of two graphs $A$ and $B$ is the graph whose vertex set is the Cartesian product $V(A\boxtimes B):=V(A)\times V(B)$ and in which two distinct vertices $(x_1,y_1)$ and $(x_2,y_2)$ are adjacent if and only if:
\begin{enumerate}
  \item  $x_1x_2 \in E(A)$ and $y_1y_2 \in E(B)$; or
  \item $x_1=x_2$ and $y_1y_2\in E(B)$; or
  \item $x_1x_2 \in E(A)$ and $y_1=y_2$.
\end{enumerate}

\begin{thm}[Dujmović \etal\ \cite{dujmovic.joret.ea:planar}]\thmlabel{product-structure}
  Every planar graph $G$ is a subgraph of a strong product $H\boxtimes P$ where $H$ is a graph of treewidth at most 8 and $P$ is a path.
\end{thm}

\thmref{product-structure} can be generalized (replacing $8$ with a larger constant) to bounded genus graphs, and more generally to apex-minor free graphs.

Dujmović, Morin, and Wood \cite{dujmovic.morin.ea:structure} gave analogous product structure theorems for some non-minor closed families of graphs including $k$-planar graphs,
powers of bounded-degree planar graphs, and $k$-nearest-neighbour graphs of points in $\R^2$. Dujmović, Esperet, Morin, Walczak, and Wood \cite{dujmovic.esperet.ea:clustered} proved that a similar product structure theorem holds for graphs of bounded degree from any (fixed) proper minor-closed class.  This is summarized in the following theorem:

\begin{thm}[\cite{dujmovic.joret.ea:planar},\cite{dujmovic.esperet.ea:clustered},\cite{dujmovic.morin.ea:structure}]\thmlabel{product-structure-all}
   Every graph $G$ in each of the following families of graphs is a subgraph of a
strong product $H\boxtimes P$ where $P$ is a path and $H$ is a graph of bounded
treewidth:

   \begin{compactitem}
   \item graphs of bounded genus and, more generally, apex-minor free graphs;
     \item bounded degree graphs that exclude a fixed graph as
a minor;
   \item $k$-planar graphs and, more generally, $(g,k)$-planar graphs. 
 \end{compactitem}
\end{thm}



\section{Bulk Trees}
\seclabel{bulk-trees}

Our labelling scheme for a subgraph $G$ of $H \boxtimes P$ uses labels that depend in part on the rows ($H$-coordinates) of $G$, where each row corresponds to one vertex of $P$: Say $P$ consists of vertices $1, 2, \dots, h$ in this order, then the \emph{$i$-th row} of $G$ is the subgraph $H_i$ of $G$ induced by the vertex set $\{(v, i) \in V(G)\}$.
A naive approach to create labels for each $H_i$ is to use a labelling scheme for bounded treewidth graphs; roughly, this entails building a specific binary search tree $T_i$ and mapping each vertex $v$ of $H_i$ onto a node $x$ of $T_i$ that we call the \emph{position} of $v$ in $T_i$.
The label of $(v, i)$ encodes the position of $v$ in $T_i$ plus some small extra information (see~\secref{hxp}).
This way, we can determine if two vertices $(v,i)$ and $(w,i)$ in the same row are adjacent.

The key problems that we face here though are queries of the type $(v,i)$ and $(w,i+1)$: We would like to determine adjacency on the $H$-coordinate using $T_i$ or $T_{i+1}$.
We could extend the node set of $T_{i+1}$ so that it represents all vertices from $H_{i}$. This way we know that both $v$ and $w$ are represented in $T_{i+1}$.
However, we still have a major issue: the label of $(v,i)$ describes the position of $v$ in $T_i$ but not in $T_{i+1}$.
In this setup, in order to determine if $v$ and $w$ are adjacent in $H$ we need to know their respective positions in the same binary search tree.
However, there is in principle no relation between the position of $v$ in $T_i$ and its position in $T_{i+1}$.

To circumvent this difficulty, we build the binary search trees $T_1,\dots,T_h$ one by one, starting with a balanced binary search tree, in such a way that $T_{i+1}$ is obtained from $T_i$ by performing carefully structured changes. By storing some small extra information related to these changes in the label of $(v,i)$, this will allow us to obtain the position of $v$ in $T_{i+1}$.
Finally, we also need to guarantee that the binary search trees in our sequence are balanced so that the labels are of length $\log|T_i|$ plus a lower-order term.

In this section, we introduce three operations on a binary search tree that will allow us to carry out this plan.
These operations are called \emph{bulk insertion}, \emph{bulk deletion}, and \emph{rebalancing}.
Starting from a perfectly balanced binary search tree $T_1$, each tree $T_i$ in our sequence $T_1, \dots, T_h$ will be obtained from $T_{i-1}$ by applying these three operations.


\subsection{Bulk Insertion}

The bulk insertion operation, $\textsc{BulkInsert}(I)$,
in which a finite set $I\subset  \R\setminus V(T)$ of new values are inserted into a binary search tree $T$, is implemented using the standard insertion algorithm for binary search trees.  For each $x\in I$, $P_T(x)$ ends at an external node $x'$ of $T^+$ whose parent $y$ is a node of $T$. We simply make $x$ a child of $y$.  Doing this for each $x\in I$ (in any order) results in a new tree $T'$ with $V(T')=V(T)\cup I$.
%
%

\begin{lem}\lemlabel{insertion-depth}
  Let $T$ be any binary search tree and let $I$ be a finite set of values from $\R\setminus V(T)$ such that
  $V(T)$ $1$-chunks $I$.
  Apply $\textsc{BulkInsert}(I)$ to $T$ to obtain $T'$.
  Then $T'$ is a supergraph of $T$ and $h(T')\le h(T)+1$.
\end{lem}

\begin{proof}
  That $T'$ is a supergraph of $T$ is obvious.  Next note that the 1-chunking property ensures that, for any $x\in I$, the parent of $x$ in $T'$ is also in $T$.  Thus any root-to-leaf path in $T'$ consists of a root-to-leaf path in $T$ followed by at most one node in $I$.  Therefore $h(T')\le h(T)+1$.
\end{proof}

\begin{lem}\lemlabel{insertion-size}
  Let $T$ be any binary search tree and let $I$ be a set of values from $\R\setminus V(T)$ such that
  $V(T)$ $1$-chunks $I$.
  Apply $\textsc{BulkInsert}(I)$ to $T$ to obtain $T'$.
  Let $x$ be any node of $T$ and let $T_x$ and $T_x'$ be the subtrees of $T$ and $T'$, respectively, rooted at $x$.
  Then $|T_x|\le |T_x'|< 4|T_x|$.
\end{lem}

\begin{proof}
We clearly have  $|T_x|\le |T_x'|$.  By definition, $V(T)$ 1-chunks $I:=V(T')\setminus V(T)$.  This implies that $V(T_x)$ 1-chunks $I_x:=V(T_x')\setminus V(T_x)$.  Therefore $|I_x|\le |T_x|+1$, so $|T_x'|=|T_x|+|I_x|\le 2|T_x|+1 \le 3|T_x| < 4|T_x|$.
\end{proof}

In \lemref{insertion-size} and in \lemref{deletion-size}, below, we use the constant $4$ rather than $3$ because this later simplifies calculations involving binary logarithms.

\subsection{Bulk Deletion}

The bulk deletion operation, $\textsc{BulkDelete(D)}$, of a subset $D$ of nodes of a binary search tree $T$ is implemented as a series of $|D|$ individual deletions, performed in any order. For each $x\in D$, the deletion of $x$ is implemented by running the following recursive algorithm:  If $x$ is a leaf, then simply remove $x$ from $T$.  Otherwise, $x$ has at least one child.  If $x$ has a left child, then recursively delete the largest value $x'$ in the subtree of $T$ rooted at the left child of $x$ and then replace $x$ with $x'$.  Otherwise $x$ has a right child, so recursively delete the smallest value $x'$ in the subtree of $T$ rooted at the right child of $x$ and then replace $x$ with $x'$.

\begin{lem}\lemlabel{deletion-signature}
  Let $T$ be any binary search tree and let $D$ be a set of values from $V(T)$.
  Apply $\textsc{BulkDelete}(D)$ to $T$ to obtain a new tree $T'$.
  Then, for any node $x$ in $T'$, $\sigma_{T'}(x)$ is a prefix of $\sigma_T(x)$.
  In particular, $h(T')\leq h(T)$.
\end{lem}

\begin{proof}
  This follows immediately from the fact the only operations performed during a bulk deletion are (i)~deletion of leaves and (ii)~using a value $x'$ to replace the value of one of its $T$-ancestors $x$.  The deletion of a leaf has no effect on $\sigma_{T'}(x)$ for any node $x$ in $T'$.  For any node $z$ in $T'$ other than $x'$, (ii) has no effect on $\sigma_T(z)$.  For the node $x'$, (ii) has the effect of replacing $\sigma_T(x')$ by its length-$d_T(x)$ prefix.
\end{proof}

\begin{lem}\lemlabel{deletion-size}
  Let $T$ be any binary search tree and let $D$ be a strict subset of $V(T)$ such that $V(T)\setminus D$ $1$-chunks $D$.
  Apply $\textsc{BulkDelete}(D)$ to $T$ to obtain a new tree $T'$.
  Then $|T|/4 \le |T'|\le |T|$.
\end{lem}

\begin{proof}
 We clearly have $|T'|\le |T|$.  Since $V(T)\setminus D$ $1$-chunks $D$, we have $|D|\le |V(T)|-|D| + 1 \le 2(|V(T)|-|D|)$ (since $D\subset V(T)$, $|V(T)|-|D|\ge 1$), so $|D|\le (2/3)|V(T)|$. Thus $|T'| \ge |T| - (2/3)|T| = |T|/3 > |T|/4$.
\end{proof}


\subsection{Rebalancing}

The rebalancing operation on a binary search tree $T$ uses several subroutines that we now discuss, beginning with the most fundamental one:  $\textsc{Split}(x)$.

\subsubsection{$\textsc{Split}(x)$}

The argument of $\textsc{Split}(x)$ is a node $x$ in $T$ and the end result of the subroutine is to split $T$ into two binary search trees $T_{<x}$ and $T_{>x}$ where $V(T_{<x})=\{z\in V(T): z<x\}$ and $V(T_{>x})=\{z\in V(T): z>x\}$. Refer to \figref{split}.  Let $P_T(x_r)=x_0,\dots,x_r$ be the path in $T$ from the root $x_0$ of $T$ to $x=x_r$.  Partition $x_0,\dots,x_{r-1}$ into two subsequences $a:=a_1,\dots,a_s$ and $b:=b_1,\dots,b_t$ where the elements of $a$ are less than $x$ and the elements of $b$ are greater than $x$.
Note that the properties of a binary search tree guarantee that
\[
a_1 < \cdots < a_s < x < b_t < \cdots < b_1.
\]
Make a binary search tree $T_0$ that has $x$ as root, the path $a_1,\dots,a_s$ as the left subtree of $x$ and the path $b_1,\dots,b_t$ as the right subtree of $x$.  Note that $a_{i+1}$ is the right child of $a_i$ for each $i\in\{1,\dots,s-1\}$ and $b_{i+1}$ is the left child of $b_i$ for each $i\in\{1,\dots,t-1\}$.

\begin{figure}
  \begin{center}
    \includegraphics{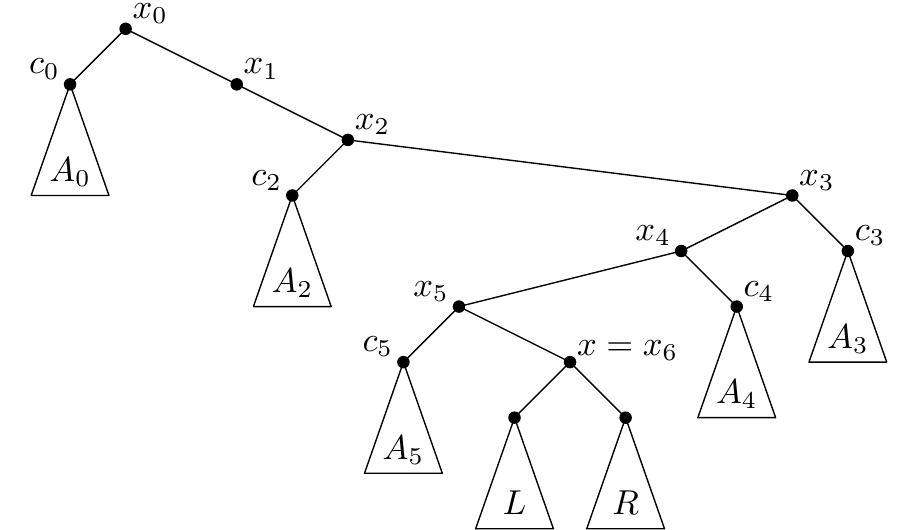} \\[1ex]
    $\Downarrow$ \\[1ex]
    \includegraphics{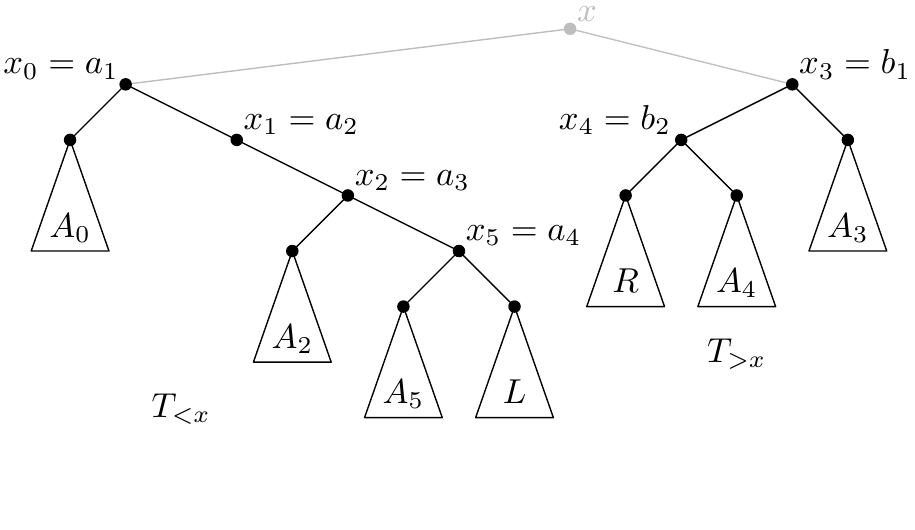}
  \end{center}
  \caption{The operation of $\textsc{Split}(x)$.}
  \figlabel{split}
\end{figure}

Next, consider the forest $F:=T-\{x_0,\dots,x_r\}$. This forest consists of $r+2$ (possibly empty) trees $A_1,\dots,A_{r-1},L,R$ where $L$ and $R$ are the subtrees of $T$ rooted at the left and right child of $x$ in $T_x$ and, for each $i\in\{1,\dots,r-1\}$, $A_i$ is the subtree of $T$ rooted at the child $c_i\neq x_{i+1}$ of $x_i$ (if such a child exists, otherwise $A_i$ is empty).  Make a binary search tree $T_x$ by replacing each of the $r+2$ external nodes of $T_0^+$ with the corresponding tree in $F$.  Finally, let $T_{<x}$ be the subtree of $T_x$ rooted at the left child of $x$ and let $T_{>x}$ be the subtree of $T_x$ rooted at the right child of $x$ in $T_x$.

\begin{lem}\lemlabel{split-height}
  Let $T$ be any binary search tree, let $x$ be any node of $T$, and apply $\textsc{Split}(x)$ to obtain $T_{<x}$ and $T_{>x}$.
  Then $h(T_{<x})\le h(T)$ and $h(T_{>x})\le h(T)$.
\end{lem}

\begin{proof}
  Note that for each node $z$ of $T_{<x}$, we have $V(P_{T_{<x}}(z))\subseteq V(P_T(z))$,
  so $d_{T_{<x}}(z)\le d_T(z)$.
  Therefore $h(T_{<x})\le h(T)$. The argument for $T_{>x}$ is symmetric.
\end{proof}

The following observation shows that there is a simple relationship between a node's signature in $T$ before calling $\textsc{Split}(x)$ and its signature in $T_{<x}$ or $T_{>x}$.

\begin{obs}\obslabel{split-signature}
  Let $T$, $x$, $x_0,\dots,x_r$, $A_1,\dots,A_{r-1},L,R$, $a_1,\dots,a_s$, and $b_1,\dots,b_t$ be defined as above. Then
  \begin{compactenum}
    \item for each $j\in\{1,\dots,s\}$ where $a_j=x_i$
    \begin{compactenum}
      \item $\sigma_{T_{<x}}(a_j)=1^{j-1}$, and
      \item $\sigma_{T_{<x}}(z) = 1^{j-1},0,\sigma_{A_i}(z)$ for each $z\in V(A_i)$;
    \end{compactenum}
    \item for each $j\in\{1,\dots,t\}$ where $b_j=x_i$
    \begin{compactenum}
      \item $\sigma_{T_{>x}}(b_j)=0^{j-1}$, and
      \item $\sigma_{T_{>x}}(z) = 0^{j-1},1,\sigma_{A_i}(z)$ for each $z\in V(A_i)$;
    \end{compactenum}
    \item $\sigma_{T_{<x}}(z)=1^s,\sigma_L(z)$ for each $z\in V(L)$; and
    \item $\sigma_{T_{>x}}(z)=0^t,\sigma_R(z)$ for each $z\in V(R)$.
  \end{compactenum}
  In particular, for any $z\in V(T)\setminus\{x\}$, $\sigma_{T_{<x}}(z)$ or $\sigma_{T_{>x}}(z)$ can be obtained from $\sigma_T(z)$ by deleting a prefix and replacing it with one of the $4\cdot h(T)$ strings in $\Pi:=\bigcup_{j=0}^{h(T)-1}\{0^j,0^j1,1^j,1^j0\}$.
\end{obs}

\subsubsection{$\textsc{MultiSplit}(x_1,\dots,x_c)$}

From the $\textsc{Split}(x)$ operation we build the $\textsc{MultiSplit}(x_1,\dots,x_c)$ operation that takes as input a sequence of nodes $x_1<\cdots<x_c$ of $T$.  For convenience, define $x_0=-\infty$ and $x_{c+1}=\infty$.  The effect of $\textsc{MultiSplit}(x_1,\dots,x_c)$ is to split $T$ into a sequence of binary search trees $T_0,\dots,T_{c}$ where, for each $i\in\{0,\dots,c\}$, $V(T_i)=\{z\in V(T): x_i< z<x_{i+1}\}$.

The implementation of $\textsc{MultiSplit}(x_1,\dots,x_c)$ is straightforward divide-and-conquer:  If $c=0$, then there is nothing to do.  Otherwise, call $\textsc{Split}(x_{\lceil c/2\rceil})$ to obtain $T_{<x_{\lceil c/2\rceil}}$ and $T_{>x_{\lceil c/2\rceil}}$.  Next, apply $\textsc{MultiSplit}(x_1,\dots,x_{\lceil c/2\rceil-1})$ to $T_{<x_{\lceil c/2\rceil}}$ to obtain $T_0,\dots,T_{\lceil c/2\rceil-1}$ and then apply $\textsc{MultiSplit}(x_{\lceil c/2\rceil+1},\dots,x_c)$ to $T_{>x_{\lceil c/2\rceil}}$ to obtain $T_{\lceil c/2\rceil},\dots,T_c$.

The following lemma is immediate from \lemref{split-height}.
\begin{lem}\lemlabel{multisplit-height}
  Let $T$ be any binary search tree and apply $\textsc{MultiSplit}(x_1,\dots,x_c)$ to $T$ to obtain $T_0,\dots,T_c$.  Then $h(T_i)\le h(T)$ for each $i\in\{0,\dots,c\}$.
\end{lem}

\subsubsection{$\textsc{Balance}(x,k)$}

The $\textsc{Balance}(x,k)$ operation operates on the subtree $T_x$ of $T$ rooted at some node $x$ in $T$.
The goal of this operation is to balance the size of all the subtrees rooted at nodes of depth $d_T(x)+k+1$ and contained in $T_x$.
Refer to \figref{balance-x}.

\begin{figure}
    \begin{center}
      \includegraphics{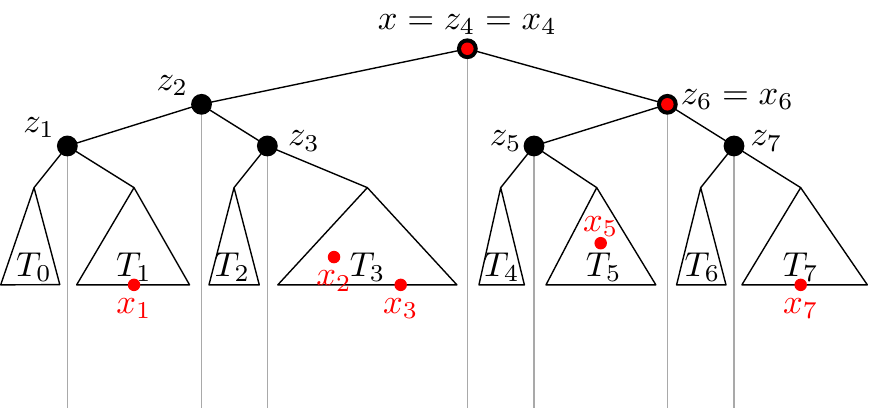} \\[-2ex]
      $\Downarrow$ \\[1ex]
      \includegraphics{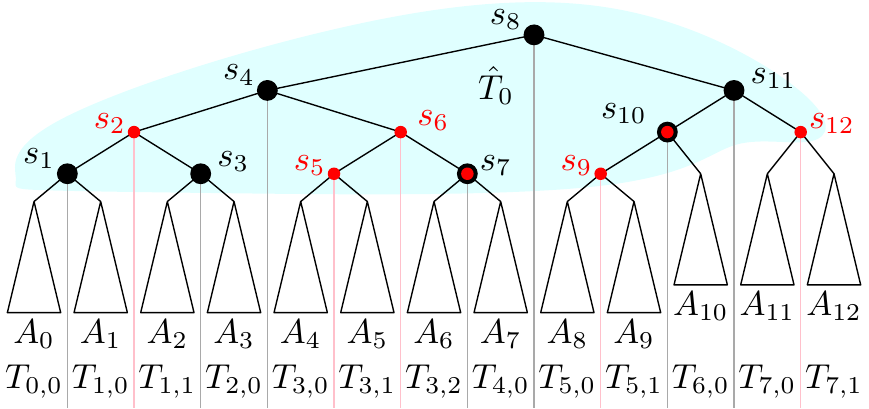}
    \end{center}
  \caption{The operation of $\textsc{Balance}(x,k)$}
  \figlabel{balance-x}
\end{figure}

If $|V(T_x)|< 2^k$, then this operation simply replaces $T_x$ with a perfectly balanced binary search tree containing $V(T_x)$.  Otherwise, let $Z:=\{z\in V(T_x): d_{T_x}(z)< k\}$.  Call the $m\le 2^k-1$ elements of $Z$  $z_1<z_2<\cdots<z_{m}$ and, for convenience, define $z_0=-\infty$ and $z_{m+1}=\infty$.

Select the nodes $X:=\{x_1,\dots,x_{2^k-1}\}$ of $T_x$ where each $x_j$ has rank $\lfloor j|V(T_x)|/2^k\rfloor$ in $V(T_x)$.\footnote{For a finite set $X\subset\R$, and $x\in\R$, the \emph{rank} of $x$ in $S$ is $|\{x'\in S: x'<x\}|$.}  The $\textsc{Balance}(x,k)$ operation will turn $T_x$ into a tree with a top part $\hat{T}_0$ that is a perfectly balanced binary search tree on $Z\cup X$.  We now describe how this is done.

$T_x-Z$ is a forest consisting of $m+1\le 2^{k}$ trees $T_0,\dots,T_m$. (Some of these trees may be empty.)  Order $T_{0},\dots,T_m$ so that, for each $i\in\{0,\dots,m\}$ and each $x'\in V(T_i)$, $z_i< x' < z_{i+1}$.  For each $i\in\{0,\dots,m\}$, let $\{x_{i,1},\dots,x_{i,c_i}\}:=X\cap V(T_i)$ where $x_{i,1}<\cdots<x_{i,c_i}$ and define $x_{i,0}:=z_i$ and $x_{i,c_i+1}:=z_{i+1}$. Note that for each $i\in\{0,\dots,m\}$, $c_i\le |X|\le 2^k-1$.

For each $i\in\{0,\dots,m\}$, apply $\textsc{MultiSplit}(x_{i,1},\dots,x_{i,c_i})$ to the tree $T_i$.  As a result of these calls, we obtain sequences of trees $T_{i,0},\dots,T_{i,c_i}$ where, for each $i\in\{0,\dots,m\}$, each $j\in\{0,\dots,c_i\}$, and each $x'\in V(T_{i,j})$, we have $x_{i,j}<x'<x_{i,j+1}$.  Note that if $c_i=0$ (i.e., if $X$ does not intersect $V(T_i)$), then the result of this call is a single tree $T_{i,0}=T_i$.
Observe that $\bigcup_{i=0}^m\bigcup_{j=0}^{c_i} V(T_{i,j}) = V(T_x)\setminus (Z\cup X)$.

Let $p := |Z\cup X|$, let $s_1<\cdots< s_p$ denote the elements of $Z\cup X$ and define $s_0 := -\infty$ and $s_{p+1} := \infty$.  For each $\ell\in\{0,\dots,p\}$, let $i_\ell:=|Z\cap \{s_1,\dots,s_\ell\}|$ and $j_\ell:= \ell - \max\{ q\in\{1,\dots,\ell\}: s_q\in Z\}$ and let $A_\ell:=T_{i_\ell,j_\ell}$.   Then, for each $\ell\in \{0,\dots,p\}$ and each $x'\in V(A_\ell)$, we have $s_\ell < x' < s_{\ell+1}$.

Now construct a perfectly balanced tree $\hat{T}_0$ with vertex set $V(\hat{T}_0):=\{s_1,\dots,s_p\}=Z\cup X$.  The tree $\hat{T}_0$ has $p+1$ external nodes $a_0,\dots,a_p$.  We obtain a new tree $T_x'$ by replacing $a_\ell$ with $A_\ell$ for each $\ell\in\{0,\dots,p\}$ in $\hat{T}_0^+$.  In the encompassing bulk tree $T$ we replace the subtree $T_x$ with $T_x'$.

\begin{lem}\lemlabel{multisplit-depth}
  Let $T$ be any binary search tree, let $x$ be any node of $T$, and apply $\textsc{Balance}(x,k)$ to $T$ to obtain a new tree $T'$.  Then $h(T')\le h(T)+1$.
\end{lem}

\begin{proof}
  Since $\textsc{Balance}(x,k)$ only affects the subtree $T_x$ rooted at $x$, it suffices to show that $h(T_x')\le h(T_x)+1$.  For each $i\in\{0,\dots, m\}$, $T_i$ is rooted at a depth-$k$ node of $T_x$, so $h(T_i)\le h(T_x)-k$. For each $\ell \in\{0,\dots,p\}$, $A_\ell$ is obtained by an application of $\textsc{MultiSplit}$ to $T_i$ for some $i\in\{0,\dots,m\}$ so, by \lemref{multisplit-height}, $h(A_\ell)\le h(T_i)\le h(T_x)-k$.  Next, $|Z\cup X|\le |Z|+|X| \le 2^k-1 + 2^k-1 < 2^{k+1}-1$ and $\hat{T}_0$ is a perfectly balanced binary search tree of size $|\hat{T}_0|=|Z\cup X|$.  Therefore $h(\hat{T}_0)\le \lfloor\log|\hat{T}_0|\rfloor\le \lfloor\log(2^{k+1}-1)\rfloor = k$.  Finally, $h(T_x')\le h(\hat{T}_0)+1 +\max \{h(A_\ell):\ell\in\{0,\dots,p\}\} \le k +1 + h(T_x) - k\le h(T_x)+1$.
\end{proof}

The following statement captures what we win after an application of $\textsc{Balance}(x,k)$ to a binary search tree.

\begin{lem}\lemlabel{balance-x-weight}
  Let $T$ be any binary search tree, let $x$ be any node of $T$, let $T_x$ be the subtree of $T$ rooted at $x$, and apply $\textsc{Balance}(x,k)$ to $T$ to obtain a new tree $T'$.  Then, for each $T'$-descendant $z$ of $x$ with $d_{T'}(z) = d_{T}(x)+k+1$, the subtree of $T'$ rooted at $z$ has size at most $|T_x|/2^k$.  \end{lem}

\begin{proof}
  Each such subtree is a subtree of $A_\ell$ for some $\ell\in\{0,\dots,p\}$. Now, $V(A_\ell)\subset (x_j,x_{j+1})$ for some $j\in\{1,\dots,2^{k-1}\}$.  The values $x_j$ and $x_{j+1}$ have ranks $\lfloor j|T_x|/2^k\rfloor$ and $\lfloor (j+1)|T_x|/2^k\rfloor$ in the set $V(T_x)$.  Therefore, $|A_\ell|\le \lfloor (j+1)|T_x|/2^k\rfloor- \lfloor j|T_x|/2^k\rfloor -1 \le |T_x|/2^k$.
\end{proof}

\subsubsection{$\textsc{BulkBalance}(\theta,k)$}

The ultimate restructuring operation in bulk trees is $\textsc{BulkBalance}(\theta,k)$.
It calls \textsc{Balance}$(x,k)$ for each node $x$ of depth $\theta$ in $T$.
(Note that this operation has no effect if there is no such node.)
The following two lemmas are immediate consequences of \lemref{multisplit-depth} and \lemref{balance-x-weight}, respectively.

\begin{lem}\lemlabel{balance-depth}
  Let $T$ be any binary search tree and apply the $\textsc{BulkBalance}(\theta,k)$ operation to obtain a new tree $T'$.  Then $h(T')\le h(T)+1$.
\end{lem}

\begin{lem}\lemlabel{balance-weight}
  Let $T$ be any binary search tree and apply the $\textsc{BulkBalance}(\theta,k)$ operation to obtain a new tree $T'$.
  Let $x$ be any node of $T$ of depth $\theta$ and let $T_x$ be the subtree of $T$ rooted at $x$.
  Then, for each $T'$-descendant $z$ of $x$ with $d_{T'}(z) = \theta+k+1$, the subtree of $T'$ rooted at $z$ has size at most $|T_x|/2^k$.
\end{lem}

\subsection{Bulk Tree Sequences}

Let $k\ge 5$ be an integer\footnote{$k \geq 5$ is a technical requirement, to make sure that some inequalities hold later on.} and let $S_0,\dots,S_q$ be a $1$-chunking sequence.
We define a \emph{one-phase $k$-bulk tree sequence} based on $S_0,\dots,S_q$ to be a sequence $T_0, \dots, T_{q}$ of binary search trees such that $T_0$ is an arbitrary binary search tree on node set $S_0$ and,  for each $y\in \{0, \dots, q-1\}$, we have $h(T_y)>y\cdot(k+1)$ and
the tree $T_{y+1}$ is obtained from $T_y$ by applying
\begin{enumerate}[label={(\roman*)}, ref={\roman*}, noitemsep]
    \item $\textsc{BulkBalance}(y\cdot(k+1),k)$, then
    \item $\textsc{BulkInsert}(I)$ with $I:=S_{y+1} \setminus S_{y}$, and finally
    \item $\textsc{BulkDelete}(D)$ with $D:=S_{y} \setminus S_{y+1}$.
\end{enumerate}
Note that $V(T_y)=S_y$ for each $y \in \{0, \dots, q\}$.
The sequence is \emph{complete} if $h(T_q)\leq q\cdot(k+1)$.

For $k\geq 5$ and a $1$-chunking sequence $S_1,\dots,S_h$, we define
a \emph{$k$-bulk tree sequence} based on $S_1,\dots,S_h$ to be a sequence $T_1, \dots, T_h$ of binary search trees satisfying:
$T_1$ is a perfectly balanced binary search tree with $V(T_1)=S_1$, and
there exist indices $h_1,h_2,\dots,h_{\ell}$ with $1=h_1 < h_2 < \cdots <h_{\ell} = h$  such that $T_{h_j}, T_{h_j+1},\dots,T_{h_{j+1}}$ is a complete one-phase $k$-bulk tree sequence based on $S_{h_j}, S_{h_j+1},\dots,S_{h_{j+1}}$
for each $j\in\{1,\dots,\ell-2\}$, and $T_{h_{\ell-1}}, T_{h_{\ell-1}+1},\dots,T_{h_{\ell}}$ is a (non-necessarily complete) one-phase $k$-bulk tree sequence based on $S_{h_{\ell-1}},S_{h_{\ell-1}+1},\dots,S_{h_{\ell}}$.

Note that if we fix the $1$-chunking sequence $S_1,\dots,S_h$, the integer $k\ge 5$, and the starting perfectly balanced binary search tree $T_1$ with $V(T_1)=S_1$, a $k$-bulk tree sequence based on $S_1,\dots,S_h$ and starting with $T_1$ exists and is unique. It is obtained by a sequence of one-phase $k$-bulk tree sequences, where we start a new one-phase sequence as soon as the current one is complete.

This will not be needed until the final sections, but it is helpful to keep in mind that we will ultimately take $k=\max\left\{5,\left\lceil\sqrt{\log n / \log\log n}\right\rceil\right\}$ when considering a $k$-bulk tree sequence built for our $n$-vertex graph $G$, so that the expression $\Oh(k+k^{-1}\log n)$ (which appears many times in what follows), is $\omega(1)$ and $o(\log n)$.

\begin{lem}\lemlabel{bulktree-old-B-properties}
Let $T_0,\dots,T_{q}$ be a one-phase $k$-bulk tree sequence.
Then, for each $y\in\{0,\dots,q\}$
\begin{compactenum}[(i)]
\item $h(T_y)\le h(T_0) + 2y$;\label{lemma-item-B1}
\item each subtree of $T_y$ rooted at a node of depth $y\cdot (k+1)$ has size at most
$|T_0|\cdot2^{-y(k-2)}$.\label{lemma-item-B2}
\end{compactenum}
\end{lem}
\begin{proof}
The proof is by induction on $y$.  For the base case $y=0$, both properties are trivial: \itemref{lemma-item-B1} asserts that $h(T_0)\le h(T_0)$ and \itemref{lemma-item-B2} asserts that the subtree of $T_0$ rooted at the root of $T_0$ has size at most $|T_0|$.

  For the inductive step, assume $y\ge 0$ and \itemref{lemma-item-B1} holds for $T_{y}$.
  In order to get $T_{y+1}$, we first apply $\textsc{BulkBalance}(y\cdot (k+1),k)$ to $T_y$ to obtain $T'$.
  By~\lemref{balance-depth}, we have $h(T') \le h(T_y)+1$.
  Let $I:=V(T_{y+1})\setminus V(T') = V(T_{y+1})\setminus V(T_y)$.
  Since $V(T_y)$ $1$-chunks $V(T_{y+1})$ we know that $V(T')$ $1$-chunks $I$.
  Next we apply $\textsc{BulkInsert}(I)$ to $T'$ to obtain $T''$.
  Thus, by~\lemref{insertion-depth} we have $h(T'') \le h(T')+1$.
  Finally, we apply $\textsc{BulkDelete}(D)$ to $T''$ and obtain $T_{y+1}$, where $D:=V(T_{y})\setminus V(T_{y+1})$.
  By~\lemref{deletion-signature}, we have $h(T_{y+1}) \le h(T'')$.
  Altogether we have
  \[
  h(T_{y+1}) \le h(T'') \le h(T') +1 \le h(T_y) +2 \le h(T_0) + 2y +2 = h(T_0) + 2(y+1).
  \]
  Thus, \itemref{lemma-item-B1} holds for $T_{y+1}$.

    Next we establish \itemref{lemma-item-B2}.
    Assume that \itemref{lemma-item-B2} holds for $T_{y}$.
    Thus, every subtree of $T_{y}$ rooted at a node of depth $y(k+1)$ has size at most $|T_{0}|\cdot2^{-y(k-2)}$.
    Again, the first step when constructing $T_{y+1}$ from $T_{y}$ is to apply $\textsc{BulkBalance}(y(k+1),k)$ to $T_{y}$.  By \lemref{balance-weight}, this results in a tree $T'$ in which every subtree rooted at a node of depth $(y+1)(k+1)$ has size at most $|T_{0}|\cdot 2^{-y(k-2)}\cdot 2^{-k}$.  The second step is to apply $\textsc{BulkInsert}(I)$ to $T'$ to obtain a new tree $T''$.
    By \lemref{insertion-size}, every subtree of $T''$ rooted at a node of depth $(y+1)(k+1)$ has size at most $|T_{0}|\cdot 2^{-y(k-2)}\cdot 4 \cdot 2^{-k}=2^{-y(k-2)+2-k}=2^{-(y+1)(k-2)}$.  Finally, the third step is to perform $\textsc{BulkDelete}(D)$ on $T''$ to obtain $T_{y+1}$.  Bulk deletion does not increase the size of any subtree, so every subtree of $T_{y+1}$ rooted at a node of depth $(y+1)(k+1)$ has size at most $|T_{0}|\cdot 2^{-(y+1)(k-2)}$, as desired.
\end{proof}

\begin{cor}\corlabel{height-y-star}
Let $T_0,\dots,T_{q}$ be a one-phase $k$-bulk tree sequence.
Then,
\[
q \leq \left\lceil\tfrac{\log|T_0|}{k-2}\right\rceil.
\]
\end{cor}
\begin{proof}
Arguing by contradiction, suppose that $q > \left\lceil\tfrac{\log|T_0|}{k-2}\right\rceil$.
Note that when we take the logarithm of the upper bound in \lemref{bulktree-old-B-properties}\itemref{lemma-item-B2} for $y:=\left\lceil\tfrac{\log|T_0|}{k-2}\right\rceil$,  we have
\[
\log|T_0| - \left\lceil\tfrac{\log|T_0|}{k-2}\right\rceil\cdot(k-2)
\leq \log|T_0|-\tfrac{\log|T_0|}{k-2}\cdot(k-2)
= 0.
\]
Thus, each subtree of $T_{y}$ rooted at a node of depth $y(k+1)$ has size at most
\[
|T_0|\cdot2^{-y(k-2)} \leq 1,
\]
and hence $h(T_y) \leq y(k+1)$, which violates the height condition in the definition of a one-phase $k$-bulk tree sequence.
\end{proof}

\begin{lem}\lemlabel{bulktree-height-i}
Let $T_0,\dots,T_{q}$ be a complete one-phase $k$-bulk tree sequence,
and let $r_0 := h(T_0)-\log|T_0|$.  Then, for each $y\in\{0,\dots,q\}$,
\begin{compactenum}[(i)]
\item $ |T_0|/4^y\le |T_y|$, and thus $\log|T_0| \le \log|T_y| + 2y$;\label{height-diff}
  \item $q=\Oh(k^{-1}\log|T_y|)$; \label{ystar-bound}
    \item $h(T_y) = \log|T_y| + r_0+\Oh(k^{-1}\log|T_y|) $; and
    \label{bulktree-height-item-i}
    \item $h(T_{q}) = \log|T_{q}|+\Oh(k+k^{-1}\log|T_{q}|)$.\label{bulktree-height-item-ii}
  \end{compactenum}
\end{lem}
\begin{proof}
Let $I_0,\dots,I_{q-1}$ and $D_0,\dots,D_{q-1}$ be the sets so that $T_{y+1}$ is obtained from $T_y$ by
rebalancing and then applying $\textsc{BulkInsert}(I_y)$ and $\textsc{BulkDelete}(D_y)$, for each $y\in\{0,\dots,q-1\}$.
First, recall that by \lemref{insertion-size} and \lemref{deletion-size} we have $|T_{y+1}|\ge |T_y|/4$.\footnote{The condition, in \lemref{deletion-size}, that $D$ is a \emph{strict} subset of $V(T)$ is satisfied since $\textsc{BulkDelete}(D_y)$ is performed on an intermediate tree $T_y'$ with $V(T_y')=V(T_y)\cup I_y$ and produces a tree $T_y''$ with $V(T_y')\setminus D_y = V(T_y'')=V(T_{y+1})=S_{y+1}$.  By definition $S_{y+1}$ is non-empty, therefore $D=D_y\neq V(T_y')$.}

Iterating this starting with $T_0$ implies that $|T_0|/4^y\le |T_y|$, and thus $\log|T_0| \le \log|T_y| + 2y$ for each $y\in\{0,\dots,q\}$, which proves (\ref{height-diff}).

\noindent By~\corref{height-y-star}, we have $q \leq \left\lceil\frac{\log|T_0|}{k-2}\right\rceil$.
Note that, for each $y\in\{0,\dots,q\}$,  we have
\begin{align}
q
&\le \tfrac{\log |T_0|}{k-2}+1\le
\tfrac{\log |T_y|+2y}{k-2}+1 \le \tfrac{\log |T_y|+2q}{k-2}+1,
&\text{(by (\ref{height-diff}), and since $y\leq q$)}\nonumber\\
\intertext{and rewriting this yields (using that $k\geq 5$)}
q
&\le \tfrac{\log |T_y|}{k-4}+ \tfrac{k-2}{k-4} \leq \tfrac{k}{k-4}\cdot\tfrac{\log |T_y|}{k} + \tfrac{k-2}{k-4}\nonumber\\
&\leq 5\cdot\tfrac{\log |T_y|}{k} + 3= \Oh(k^{-1}\log|T_y|), \nonumber
\intertext{which proves (\ref{ystar-bound}). Now (\ref{bulktree-height-item-i}) follows as for each $y\in\{0,\dots,q\}$, we have}
h(T_y)
& \le h(T_0) + 2y \nonumber
&\text{(by~\lemref{bulktree-old-B-properties}\itemref{lemma-item-B1})}\\
&= \log|T_0| + 2y + r_0  \nonumber
&\text{(by definition of $r_0$)}\\
&\le \log |T_y| + 4y + r_0
&\text{(by (\ref{height-diff}))} \nonumber\\
&\le \log |T_y| + 4q + r_0
&\text{(since $y\le q$)}\nonumber\\
&= \log |T_y| + \Oh(k^{-1}\log|T_y|) + r_0.
&\text{(by (\ref{ystar-bound}))}\nonumber\\
\intertext{(\ref{bulktree-height-item-ii}) follows from}
h(T_{q})
& \le (k+1)q
&\text{(since the sequence is complete)}\nonumber\\
&\le (k+1)\left(\tfrac{\log|T_0|}{k-2} + 1\right)\nonumber
&\text{(by \corref{height-y-star})} \\
& = \left(\tfrac{k+1}{k-2}\right)\cdot\log|T_0| + (k+1)\nonumber\\
& \le \left(\tfrac{k+1}{k-2}\right)\cdot(\log|T_{q}| + 2q) + (k+1)
&\text{(by (\ref{height-diff}))}\nonumber\\
& = \log|T_{q}| + \tfrac{3}{k-2}\log|T_{q}| + 2\cdot\tfrac{k+1}{k-2}q + k+1\nonumber\\
& = \log|T_{q}| + \Oh(k+k^{-1}\log |T_{q}|).
&\text{(as $\tfrac{k+1}{k-2}\leq\tfrac{6}{3}$, and by (\ref{ystar-bound}))}\nonumber
\end{align}
\end{proof}

The following lemma shows that trees in a bulk tree sequence are balanced at all times:

\begin{lem}\lemlabel{bulk-tree-height}
  Let $T_1,\dots,T_h$ be a $k$-bulk tree sequence and let $y\in\{1,\dots,h\}$.
  Then
\[
h(T_y)\le \log|T_y| + \Oh(k+k^{-1}\log|T_y|).
\]
\end{lem}

\begin{proof}
  By the definition of a $k$-bulk tree sequence, $T_1$ is a perfectly balanced binary tree so $h(T_1)=\lceil\log|T_1|\rceil$ and the statement is satisfied for $y=1$.
  Let $h_1,h_2,\dots,h_{\ell}$ be indices with $1 = h_1 < h_2 < \cdots < h_{\ell} = h$ such that $T_{h_j},\dots,T_{h_{j+1}}$ is a complete one-phase $k$-bulk tree sequence for each $j\in\{1,\dots,\ell-2\}$, and $T_{h_{\ell-1}},\dots,T_{h_{\ell}}$ is a one-phase $k$-bulk tree sequence.
  Let $Y := \{h_1,h_2,\dots,h_{\ell-1}\}$.
  For $y\in Y\setminus\{1\}$,
  \lemref{bulktree-height-i}(\ref{bulktree-height-item-ii}) implies that $T_y$ satisfies the conditions of the lemma.

  All that remains is to show that the conditions of the lemma are satisfied for each $y\in\{1,\dots,h\}\setminus Y$. To show this, let $y_0 := \max\{ y'\in Y: y'<y \}$.  That is, $T_{y_0}$ is the tree that began the one-phase $k$-bulk tree sequence in which $T_y$ takes part.
  In this case, \lemref{bulktree-height-i}(\ref{bulktree-height-item-i}) implies that
  \[  h(T_y) \le \log |T_y| + \Oh(k^{-1}\log |T_y|) + h(T_{y_0})-\log|T_{y_0}|.\]
  Thus, all that is required is to show that $r_0 := h(T_{y_0})-\log|T_{y_0}|\in \Oh(k+k^{-1}\log|T_y|)$ so that is what we do.
  Note that by~\lemref{bulktree-height-i}(\ref{ystar-bound})
  we have $y-y_0 = \Oh(k^{-1}\log|T_y|)$.
  \begin{align*}
    r_0 &= h(T_{y_0})-\log|T_{y_0}| \\
       &= \Oh(k + k^{-1}\log|T_{y_0}|)
        & \text{(by \lemref{bulktree-height-i}(\ref{bulktree-height-item-ii}))}\\
       &= \Oh(k + k^{-1}(\log|T_{y}| + 2(y-y_0)))
        & \text{(by \lemref{bulktree-height-i}(\ref{height-diff}))} \\
       &= \Oh(k + k^{-1}\log|T_{y}| + k^{-2}\log|T_{y}|) & \text{(by \lemref{bulktree-height-i}(\ref{ystar-bound}))} \\
       &= \Oh(k + k^{-1}\log|T_{y}|).  &  & \qedhere
  \end{align*}
\end{proof}

\subsection{Transition Codes for Nodes}
\seclabel{node-transitions}

We now arrive at the \emph{raison d'être} of bulk tree sequences:  For two consecutive trees $T_y$ and $T_{y+1}$ in a bulk tree sequence and any $z\in V(T_y)\cap V(T_{y+1})$, the signatures $\sigma_{T_y}(z)$ and $\sigma_{T_{y+1}}(z)$ are so closely related that $\sigma_{T_{y+1}}(z)$ can be derived from $\sigma_{T_y}(z)$ and a short \emph{transition code} $\nu_y(z)$.  The following two lemmas make this precise.

\begin{lem}\lemlabel{node-bal-transitions}
  There exists a function $B:(\{0,1\}^*)^2\to\{0,1\}^*$ such that,
  for every binary search tree $T$, for any integers $\theta$ and $k$ with $1\leq \theta \leq h(T)$ and $k\geq1$, the following holds.
  Let $T'$ be the binary search tree obtained by an application of $\textsc{BulkBalance}(\theta,k)$ to $T$.
  For each $z\in V(T)$, there exists $\nu(z)\in\{0,1\}^*$ with $|\nu(z)| = \Oh(k\log h(T))$ such that $B(\sigma_{T}(z), \nu(z)) = \sigma_{T'}(z)$.
\end{lem}

\begin{proof}
  Recall that by~\lemref{balance-depth}, $h(T') \leq h(T)+1$.
  Recall also that $\gamma:\mathbb{N}\to\{0,1\}^*$ is a prefix-free encoding of the natural numbers such that $|\gamma(i)|=\Oh(\log i)$, for every natural number $i$ as in~\lemref{elias}.

  $\textsc{BulkBalance}(\theta,k)$ calls $\textsc{Balance}(x,k)$ for each node $x$ of depth $\theta$  in $T$.
  Recall that the changes caused by $\textsc{Balance}(x,k)$ are limited to the subtree of $T$ rooted at $x$.
  Thus, $\sigma_{T}(z)$ can be affected by $\textsc{Balance}(x,k)$ only if $x$ is a $T$-ancestor of $z$.
  In particular when $|\sigma_{T}(z)| < \theta$, the signature of $z$ does not change, that is, $\sigma_{T}(z)=\sigma_{T'}(z)$. In this case,  we define
  \[
  \nu(z):=\gamma(\theta).
  \]
  Note that in this case $|\nu(z)|=\Oh(\log\theta) = \Oh(\log h(T))$.

  Assume now that $|\sigma_{T}(z)| \geq \theta$ and let $x$ be the $T$-ancestor of $z$ at depth $\theta$.
  Recall that the application of $\textsc{Balance}(x,k)$ first identifies two sets of nodes $Z$ and $X$ that eventually form a perfectly balanced tree $\hat{T}_0$ of height at most $k$ which forms the top part of the subtree that replaces the subtree of $x$ in $T$.
  This means that if $z\in Z\cup X$ then $\sigma_{T'}(z) = \sigma_{T}(x),\sigma_{\hat{T_0}}(z)$.
  In this case, we define
  \[
  \nu(z):=\gamma(\theta),0,\gamma(|\sigma_{\hat{T_0}}(z)|),\sigma_{\hat{T_0}}(z).
  \]
  Note that in this case $|\nu(z)|=\Oh(\log\theta) + \Oh(\log k) + \Oh(k) = \Oh(\log h(T) + k)$.

  Now we are left with the case that $|\sigma_{T}(z)|\geq \theta$ and $z\not\in Z\cup X$. In particular, the node $z$ is in some tree $T_i$ of the forest $T_x - Z$ where $T_x$ is the subtree of $T$ rooted at $x$.
  (Further on we reuse the notations $T_i$, $T_{i,0},\dots,T_{i,c_i}$, etc.\ introduced in the definition of $\textsc{Balance}(x,k)$.)
  Recall that $\textsc{Balance}(x,k)$ calls $\textsc{MultiSplit}(x_{i,1},\dots,x_{i,c_i})$ on $T_i$ to obtain a sequence of trees $T_{i,0},\dots,T_{i,c_i}$ and the node $z$ ends up in one of these trees, say in $T_{i,a}$.
  Note that
  \begin{compactenum}[(i)]
  \item $\sigma_{T_i}(z)$ is a suffix of $\sigma_{T}(z)$;
  \item the application of $\textsc{MultiSplit}(x_{i,1},\dots,x_{i,c_i})$ to $T_i$ calls $\textsc{Split}(x_{
  \lceil c_i/2\rceil})$ (given $c_i>0$) to obtain two trees $T_{<x_{\lceil c_i/2\rceil}}$ and $T_{>x_{\lceil c_i/2\rceil}}$, and then recursively calls $\textsc{MultiSplit}(x_{i,1},\dots,x_{i,\lceil c_i/2\rceil-1})$  on $T_{<x_{i,\lceil c_i/2\rceil}}$ and $\textsc{MultiSplit}(x_{i,\lceil c_i/2\rceil+1},\dots,x_{i,c_i})$  on $T_{>x_{i,\lceil c_i/2\rceil}}$;
  the node $z$ lies in one of the trees $T_{<x_{i,\lceil c_i/2\rceil}}$, $T_{>x_{i,\lceil c_i/2\rceil}}$ and by~\obsref{split-signature} the signature of $z$ in the new tree can be obtained from $\sigma_{T_i}(z)$ by deleting a prefix and replacing it with one of the $4h(T_i)$ strings in $\bigcup_{j=0}^{h(T_i)-1}\{0^j,0^j1,1^j,1^j0\}$;
  \item the application of $\textsc{MultiSplit}(x_{i,1},\dots,x_{i,c_i})$ thus defines a sequence of trees starting with $T_i$ and ending with $T_{i,a}$ that all contain $z$;
  by~\lemref{split-height} the height of each of these trees is at most $h(T_i)$;
  the signature of $z$ in these trees, which is $\sigma_{T_i}(z)$ at the beginning, undergoes at most $1+\log c_i \leq 1+k$ changes
  before becoming $\sigma_{T_{i,a}}(z)$; let $b$ denote the number of these changes and, for each $j\in\{1,\dots,b\}$, let $d_j$ be the length of the prefix of the signature being deleted during the $j$-th change and let $q_j\in \{1, \dots, 4h(T_i)\}$ be a number identifying the string that this prefix is replaced with during the $j$-th change (here we use that all the trees in the sequence have height at most $h(T_i)$).
  \end{compactenum}

  Finally, $\textsc{Balance}(x,k)$ replaces the external nodes of $\hat{T_0}$ with the trees output by \textsc{MultiSplit}$(x_{i,1},\dots,x_{i,c_i})$.
  Let $z'$ be the external node of $\hat{T_0}$ that is replaced with $T_{i,a}$.
  Therefore, the signature of $z$ in $T'$ is the concatenation of $\sigma_{T}(x)$, $\sigma_{\hat{T_0}}(z')$, and $\sigma_{T_{i,a}}(z)$.
  We define $\nu(z)$ in this case as follows.
    \[
  \nu(z) := \gamma(\theta),1,\gamma(k),\gamma(|\sigma_{\hat{T_0}}(z')|),\sigma_{\hat{T_0}}(z'),\gamma(b),
  \gamma(d_1),\gamma(q_1),\dots,\gamma(d_b),\gamma(q_b).
  \]
  Note that in this case $|\nu(z)|=\Oh(\log\theta) + \Oh(\log k) + \Oh(k) + \Oh(\log k) + \Oh(2k\cdot\log h(T))=\Oh(k\log h(T))$.
  This completes the definition of $\nu(z)$.

  The function $B$ is defined as expected:
  Given $\sigma_{T}(z)$ and $\nu(z)$, the function $B$ first decodes $\theta$ and checks whether $|\sigma_{T}(z)| < \theta$. If this is the case then the signatures of $z$ in $T$ and $T''$ are the same, so $B$ outputs $\sigma_{T}(z)$.
  If $|\sigma_{T}(z)| \geq \theta$ then $B$ reads the next bit of $\nu(z)$.
  If it is $0$ then this corresponds to the case where $z\in Z\cup X$ described above, and the information encoded after is enough to recover and output $\sigma_{T'}(z)=\sigma_{T}(x),\sigma_{\hat{T_0}}(z)$.
  If the bit under consideration was $1$, then this corresponds to the case where $z\not\in Z\cup X$.
  Recall that the set $Z$ are just all nodes in $T_x$ of depths less than $k$.
  Thus, just removing first $\theta+k$ bits of $\sigma_{T}(z)$, the function $B$ obtains $\sigma_{T_i}(z)$ where $T_i$ is the aforementioned subtree of $T_x-Z$ containing $z$.
  The function $B$ then reads the rest of the information $\nu(z)$,
  which allows it to follow all the changes made to $\sigma_{T_i}(z)$ until it becomes $\sigma_{T_{i,a}}(z)$ of the signature that are described above.
  This way again $B$ computes and outputs $\sigma_{T'}(z) = \sigma_{T}(x),\sigma_{\hat{T_0}}(z'),\sigma_{T_{i,a}}(z)$.
  This completes the definition of $B$ and the proof of the lemma.
\end{proof}

\begin{lem}\lemlabel{node-transitions}
  There exists a function $B':(\{0,1\}^*)^2\to\{0,1\}^*$ such that, for each $k$-bulk tree sequence $T_1,\dots,T_h$, each $y\in\{1,\dots,h-1\}$, and each $z\in V(T_y)\cap V(T_{y+1})$, there exists $\nu_y(z)\in\{0,1\}^*$ with $|\nu_y(z)| = \Oh(k\log h(T_y))$ such that $B'(\sigma_{T_y}(z), \nu_y(z)) = \sigma_{T_{y+1}}(z)$.
\end{lem}



\begin{proof}
  Let $T_1,\dots,T_h$ be a $k$-bulk tree sequence and let $y\in\{1,\dots,h-1\}$. Let $I := V(T_{y+1})\setminus V(T_y)$ and $D := V(T_y) \setminus V(T_{y+1})$.
  Recall that the transformation of $T_{y}$ into $T_{y+1}$ occurs in three steps:
  applying $\textsc{BulkBalance}(\theta,k)$ to $T_y$ with the appropriate value of $\theta$ to obtain $T'$,
  applying $\textsc{BulkInsert}(I)$ to $T'$ to obtain $T''$, and
  applying $\textsc{BulkDelete}(D)$ to $T''$ to obtain $T_{y+1}$.
  Recall that whenever $\textsc{BulkBalance}(\theta,k)$ is applied in this context we have $\theta < h(T_y)$.

  By \lemref{deletion-signature}, \lemref{insertion-depth}, and \lemref{balance-depth} we have  $h(T_{y+1}) \leq h(T'') \leq h(T')+2 \leq h(T_y)+3$.
  Thus the heights of all these trees are $\Oh(h(T_y))$.


  Given a node $z$ appearing in both $T_y$ and $T_{y+1}$ we are going to describe $\nu_y(z)$.
  The transition code $\nu_y(z)$ consists of two parts $\nu_y^{\textsc{Bal}}(z)$ and $\nu_y^{\textsc{Del}}(z)$ devoted to different steps of the transformation from $T_y$ to $T_{y+1}$.

  The first part $\nu_y^{\textsc{Bal}}(z)$ is simply defined as
  \[
  \nu_y^{\textsc{Bal}}(z) := \gamma(|\nu(z)|),\nu(z) \enspace ,
  \]
  where $\nu(z)$ is given by an application of~\lemref{node-bal-transitions} with $T=T_y$.

  Recall that the bulk insertion of new nodes in $T'$ does not affect the signature of existing nodes in the tree, since new nodes are inserted at the leaves of $T'$.

  We next describe $\nu^{\textsc{Del}}_y(z)$ that serves to reconstruct $\sigma_{T_{y+1}}(z)$ from $\sigma_{T''}(z)$.
  This turns out to be fairly easy.
  By~\lemref{deletion-signature} we have that $\sigma_{T_{y+1}}(z)$ is just a prefix of $\sigma_{T''}(z)$.
  Therefore, it is enough to define
  \[
  \nu^{\textsc{Del}}_y(z):= \gamma(|\sigma_{T_{y+1}}(z)|).
  \]

  Finally, we define $\nu_y(z)$ to be the concatenation of $\nu^{\textsc{Bal}}_y(z)$ and $\nu^{\textsc{Del}}_y(z)$.
  It follows that $|\nu_y(z)| = \Oh(k\log h(T_y))$.

  The function $B'$ is defined as expected:
  Given $\sigma_{T_{y}}(z)$ and $\nu_y(z)$, the function $B'$ first decodes $\nu(z)$ and computes
  $B(\sigma_{T_{y}}(z),\nu(z))=\sigma_{T'}(z)=\sigma_{T''}(z)$.
  Then $B'$ decodes $|\sigma_{T_{y+1}}(z)|$ and computes a prefix of this size of $\sigma_{T''}(z)$.
  As we have seen, the prefix is $\sigma_{T_{y+1}}(z)$, which is output by $B'$.
\end{proof}

\section{Subgraphs of $P\boxtimes P$}
\seclabel{pxp}

Before continuing, we show that using the techniques developed thus far, we can already solve a non-trivial special case.  In particular, we consider the case in which $G$ is an $n$-vertex subgraph of $P_1\boxtimes P_2$ where $P_1$ is a path on $m$ vertices and $P_2$ is a path on $h$ vertices.  Thus, we identify each vertex of $G$ with a point $(x,y)\in\{1,\dots,m\}\times \{1,\dots,h\}$ in the $m\times h$ grid with diagonals, and $G$ is just a subgraph of this grid, see \figref{pxp}.
Obviously, we may assume that $m\leq n$ and $h\leq n$.

Our motivation for considering this special case is expository: The vertices of $P_1$ are integers $1,\dots,m$ that can be stored directly in a binary search tree. This makes it easier to understand the role that bulk tree sequences play in our solution.  The extension of this solution to subgraphs of $H\boxtimes P$, which is the topic of \secref{hxp}, uses exactly the same ideas but requires another level of indirection since there is no natural mapping from the vertices of $H$ onto real numbers.

\begin{figure}
  \begin{center}
    \includegraphics{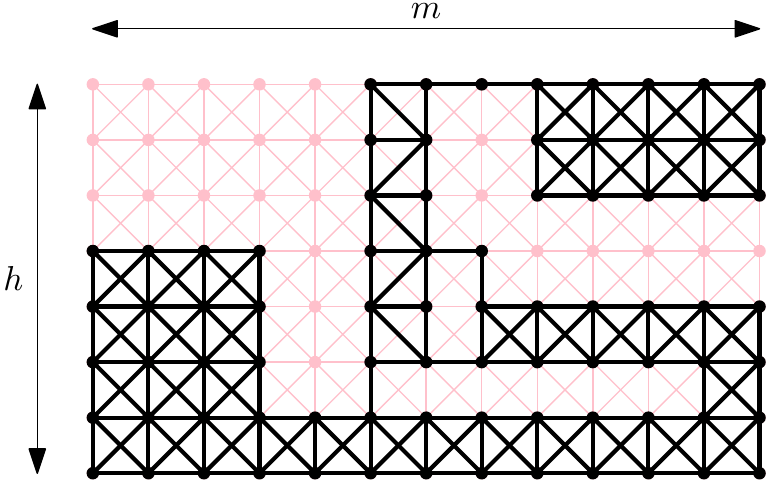}
  \end{center}
  \caption{The special case where $G$ is a subgraph of $P_1\boxtimes P_2$.}
  \figlabel{pxp}
\end{figure}

\subsection{The Labels}

For each $y\in\{1,\dots,h\}$, we let
\begin{align*}
L_y&=\{x:(x,y)\in V(G)\}, \textrm{ and}\\
L^+_y&=L_y\cup\{x-1:(x,y)\in V(G)\}.
\end{align*}
Note that $\sum_{y=1}^h |L_y| = n$ and $\sum_{y=1}^h |L^+_y|\le 2n$.
Let $L^+_0 := \emptyset$.

Let $V_1,\dots,V_{h}$ be the $1$-chunking sequence obtained by applying~\lemref{fractional} to the sequence $L^+_{1}\cup L^+_{0}, \dots, L^+_{h}\cup L^+_{h-1}$. Thus for each $y\in\{1,\dots,h\}$, we have
\begin{align*}
V_y&\supseteq L^+_{y}\cup L^+_{y-1}, \textrm{ and}\\
\textstyle\sum_{y=1}^h |V_y|&\le 4\textstyle\sum_{y=1}^h |L^+_{y}\cup L^+_{y-1}|\leq 16n.
\end{align*}
Next, let $T_1,\dots,T_h$ be a $k$-bulk tree sequence based on $V_1,\dots,V_{h}$ (recall that if we fix the starting perfectly balanced binary search tree $T_1$ with vertex set $V_1$, this sequence exists and is unique).
We discuss the asymptotically optimal choice for the value of $k$ at the end of the section.
By~\lemref{bulk-tree-height}, for each $y\in\{1,\dots,h\}$, we have
\begin{align*}
h(T_y)&=\log|T_y| + \Oh(k+k^{-1}\log |T_y|)\\
&\leq\log|T_y| + \Oh(k+k^{-1}\log n).
\end{align*}

Let $A:(\{0,1\}^{*})^2\to\{0,1\}^*$ be the function, given by \lemref{row-code} such that
using the weight function $w(y):=|T_y|$ for each $y\in\{1,\dots,h\}$,
we have a prefix-free code $\alpha:\{1,\dots,h\}\to\{0,1\}^*$ such that
\begin{align*}
|\alpha(y)|&=\log\left(\textstyle\sum_{i=1}^h|T_i|\right) - \log|T_y| + \Oh(\log\log h)\\
&\leq \log n - \log|T_y| + \Oh(\log\log n),
\end{align*}
for each $y\in\{1,\dots,h\}$, and $A(\alpha(i),\alpha(j))$ outputs $0$, $1$, $-1$, or $\perp$, depending whether the value of $j$ is $i$, $i+1$, $i-1$, or some other value, respectively.

Let $B':(\{0,1\}^{*})^2\to\{0,1\}^*$ be the function, given by \lemref{node-transitions}, such that
for each $y\in\{1,\dots,h-1\}$ and each $x\in L_{y}\subseteq V(T_y)\cap V(T_{y+1})$,
there exists a code $\nu_{y}(x)$ with $|\nu_{y}(x)|=\Oh(k\log h(T_{y}))=\Oh(k\log\log n+k \log k)$ such that $B'(\sigma_{T_{y}}(x),\nu_{y}(x))=\sigma_{T_{y+1}}(x)$.

Let $D:(\{0,1\}^{*})^2\to\{0,1\}^*$ be the function, given by \obsref{predecessor-encoding}, such that for every binary search tree $T$, and every $i$ such that $i-1$ and $i$ are in $T$, there exists $\delta_T(i) \in \{0,1\}^*$ with $|\delta_T(i)|=\Oh(\log h(T))$ such that
$D(\sigma_T(i),\delta_T(i))=\sigma_T(i-1)$.

Finally, given a vertex $v=(x,y)$ of $G$, we define
an array $a(v)$ of $8$ bits indicating whether each of the edges between $(x,y)$ and $(x\pm 1,y\pm 1)$ are present in $G$. Note that some of these $8$ vertices may not even be present in $G$ in which case the resulting bit is set to $0$ since the edge is not present in $G$.

Now, in the labelling scheme for $G$, each vertex $v=(x,y)\in V(G)$ receives a label that is the concatenation of the following bitstrings:
\begin{compactenum}[(P1)]
\item $\alpha(y)$;
\item $\gamma(|\sigma_{T_y}(x)|)$, $\sigma_{T_y}(x)$;
\item $\delta_{T_y}(x)$;
\item
if $y\neq h$ then $1,\delta_{T_{y+1}}(x)$;\\
if $y=h$ then $0$;
\item
if $y\neq h$ then $1,\nu_y(x)$;\\
if $y=h$ then $0$; and
\item $a(v)$.
\end{compactenum}
Two major components of this label are $\alpha(y)$, of length $\log n - \log|T_y| + \Oh(\log\log n)$,
and $\sigma_{T_y}(x)$, of length $\log|T_y| + \Oh(k+k^{-1}\log n)$.
Together they have length $\log n + \Oh(k+k^{-1}\log n+\log\log n)$.
The lengths of the remaining components are as follows:
$\gamma(|\sigma_{T_y}(x)|)$,
$\delta_{T_y}(x)$, and $\delta_{T_{y+1}}(x)$
 have length $\Oh(\log\log n + \log k)$, $\nu_y(x)$ has length $\Oh(k\log\log n + k\log k)$, and $a(v)$ has length $\Oh(1)$.
Thus, in total the label has length  $\log n+ \Oh(k\log\log n + k\log k + k^{-1}\log n)$.

\subsection{Adjacency Testing}

First note that from a given label of $v=(x,y) \in V(G)$, we can decode each block of the label.
This is because $\alpha(y)$ is prefix-free, $\gamma(|\sigma_{T_y}(x)|)$ is prefix-free so when we read it we know how long is $\sigma_{T_y}(x)$ and we can isolate it as well.
The $\delta$-codes are prefix-free again and $\nu_y(x)$ can be decoded as outlined in the proof of \lemref{node-transitions}.
Finally, the last $8$-bits correspond to $a(v)$.

Given the labels of two vertices $v_1=(x_1,y_1)$ and $v_2=(x_2,y_2)$ in $G$ we can test if they are adjacent as follows.

Looking up the value of $A(\alpha(y_1),\alpha(y_2))$, we determine which of the following applies:
\begin{enumerate}
  \item $|y_1-y_2|\ge 2$: In this case we immediately conclude that $v_1$ and $v_2$ are not adjacent in $G$ since they are not adjacent even in $P_1\boxtimes P_2$. 

  \item $y_1=y_2$: In this case, let $y:=y_1=y_2$.
  If the two bitstrings $\sigma_{T_y}(x_1)$, $\sigma_{T_y}(x_2)$ are the same,
  we conclude that $x_1=x_2$ and $y_1=y_2$, so $v_1=v_2$ and we should output that they are not adjacent.
  Otherwise, we
  lexicographically compare $\sigma_{T_y}(x_1)$ and $\sigma_{T_y}(x_2)$.
  Without loss of generality, $\sigma_{T_y}(x_1)$ is smaller than $\sigma_{T_y}(x_2)$.
  Therefore, by \obsref{lexicographic}, $x_1<x_2$.
  Recall that $x_2 \in L_y$ and $L^+_y\subseteq V(T_y)$, so $x_2-1 \in V(T_y)$.
  We compute $D(\sigma_{T_y}(x_2),\delta_{T_y}(x_2))=\sigma_{T_y}(x_2-1)$. If $\sigma_{T_y}(x_2-1)\neq \sigma_{T_y}(x_1)$, then we immediately conclude that $x_2 < x_1-1$, so $v_1$ and $v_2$ are not adjacent in $G$, since they are not adjacent even in $P_1\boxtimes P_2$.  Otherwise, we know that
  $v_1=(x_2-1,y)$ and $v_2=(x_2,y)$ are adjacent in $P_1\boxtimes P_2$.
  Now we use the relevant bit of $a(v_1)$ (or $a(v_2)$) to determine if $v_1$ and $v_2$ are adjacent in $G$.

  \item $y_1=y_2-1$:
  In this case, we compute $B'(\sigma_{T_{y_1}}(x_1), \nu_{y_1}(x_1))=\sigma_{T_{y_2}}(x_1)$.
  Let $y:=y_2$.
  If the two bitstrings $\sigma_{T_y}(x_1)$, $\sigma_{T_y}(x_2)$ are the same,
  we conclude that $x_1=x_2$. Thus $v_1=(x_1,y-1)$ and $v_2=(x_1,y)$ are adjacent in $P_1\boxtimes P_2$.
  Now we look up the relevant bit of $a(v_1)$ (or $a(v_2)$) to determine
  if $v_1$ and $v_2$ are adjacent in $G$.
  Otherwise, we lexicographically compare $\sigma_{T_y}(x_1)$ and $\sigma_{T_y}(x_2)$.
  If $\sigma_{T_y}(x_1)$ is smaller than $\sigma_{T_y}(x_2)$, then we conclude that $x_1<x_2$.
  Recall that $x_2 \in L_y$ and $L^+_y\subseteq V(T_y)$, so $x_2-1 \in V(T_y)$.
  We compute $D(\sigma_{T_y}(x_2),\delta_{T_y}(x_2))=\sigma_{T_y}(x_2-1)$.
  If $\sigma_{T_y}(x_2-1)\neq \sigma_{T_y}(x_1)$, then we immediately conclude that $v_1$ and $v_2$ are not adjacent in $G$, since they are not adjacent even in $P_1\boxtimes P_2$.  Otherwise, we know that
  $v_1=(x_2-1,y-1)$ and $v_2=(x_2,y)$ are adjacent in $P_1\boxtimes P_2$.
  Now we use the relevant bit of $a(v_1)$ (or $a(v_2)$) to determine if $v_1$ and $v_2$ are adjacent in $G$.
  If $\sigma_{T_y}(x_1)$ is larger than $\sigma_{T_y}(x_2)$, then we conclude that $x_1>x_2$.
  Recall that $x_1 \in L_{y-1}$ and $L^+_{y-1}\subseteq V(T_y)$, so $x_1-1 \in V(T_y)$.
  We compute $D(\sigma_{T_{y}}(x_1),\delta_{T_{y}}(x_1))=\sigma_{T_y}(x_1-1)$.
  If $\sigma_{T_y}(x_1-1)\neq \sigma_{T_y}(x_2)$, then we immediately conclude that $v_1$ and $v_2$ are not adjacent in $G$, since they are not adjacent even in $P_1\boxtimes P_2$.
  Otherwise, we know that
  $v_1=(x_1,y-1)$ and $v_2=(x_1-1,y)$ are adjacent in $P_1\boxtimes P_2$.
  Now we use the relevant bit of $a(v_1)$ (or $a(v_2)$) to determine if $v_1$ and $v_2$ are adjacent in $G$.
  \item $y_2=y_1-1$: In this case, we compute $B'(\sigma_{T_{y_2}}(x_2),\nu_{y_2}(x_2))=\sigma_{T_{y_1}}(x_2)$.  Now we proceed as in the previous case with the roles of $v_1$ and $v_2$ reversed.
\end{enumerate}

This establishes our first result:

\begin{thm}\thmlabel{pxp}
  The family $\mathcal{G}$ of $n$-vertex subgraphs of a strong product $P\boxtimes P$ where $P$ is a path has a $(1+o(1))\log n$-bit adjacency labelling scheme.
\end{thm}

\begin{rem}
  The $o(\log n)$ term in the label length of \thmref{pxp} is $\Oh(k\log\log n + k \log k + k^{-1}\log n)$.  An asymptotically optimal choice of $k$ is therefore $k=\max\left\{5,\left\lceil\sqrt{\log n / \log\log n}\right\rceil\right\}$, yielding labels of length $\log n + \Oh\left(\sqrt{\log n\log\log n}\right)$.
\end{rem}


\section{Subgraphs of $H\boxtimes P$}
\seclabel{hxp}

In this section we describe adjacency labelling schemes for graphs $G$ that are subgraphs of $H\boxtimes P$ where $H$ is a graph of treewidth $t$ and $P$ is a path.

Let $t$ be a positive integer.
A graph $H$ is a \emph{$t$-tree} if there is an ordering $v_1,\dots,v_m$ of $V(H)$ such that
for every $i\in\{1,\dots,m\}$, the neighbors of $v_i$ earlier in the order, i.e., $N_H(v_i) \cap \{v_1,\dots,v_{i-1}\}$ induce a clique of size at most $t$ in $H$.
(Let us emphasize that this is slightly more general than the usual definition of $t$-trees from the literature, which requires the neighbors of $v_i$ earlier in the order to be a clique of size exactly $\min\{i-1, t\}$; this broader definition will be more convenient for our purposes.)
A vertex-ordering witnessing that $H$ is a $t$-tree is called an \emph{elimination ordering}. Note that if $H$ is a $t$-tree with a given elimination ordering, then for any subset $X$ of vertices of $H$, the subgraph $H[X]$ of $H$ induced by $X$ is a $t$-tree and the restriction of the elimination ordering of $H$ to $X$ is an elimination ordering of  $H[X]$.
Every graph of treewidth $t$ is a spanning subgraph of a $t$-tree.
For this reason, we may restrict ourselves to the case $H\boxtimes P$ where $H$ is a $t$-tree, which we do.

Given a $t$-tree $H$, we fix an elimination ordering $v_1,\dots,v_m$.
For every $i\in\{1,\dots,m\}$, the \emph{family clique} $C_H(v_i)$ is defined as
$N_H[v_i]\cap \{v_1,\dots,v_{i}\}$.
Note that $v_i \in C_H(v_i)$.

\subsection{$t$-Trees and Interval Graphs}

The \emph{clique number} $\omega(G)$ of a graph $G$ is the maximum size of a clique in $G$. The closed real interval with endpoints $a<b$ is denoted by $[a,b]$.
For a finite set $S$ of intervals,
the \emph{interval intersection graph} $G_S$ of $S$ is the graph
with vertex set $V(G_S):=S$ and in which there is an edge between two distinct intervals
if and only if the intervals intersect.

The following well-known result states that every $m$-vertex $t$-tree is a subgraph of an interval graph of clique number $\Oh(t\log m)$.\footnote{The specific value $\log_3(2m+1) + 1$ in \lemref{interval-representation} is obtained by applying a result of Scheffler \cite{scheffler:optimal} on the tree underlying the width-$t$ tree decomposition of $H$.}

\begin{lem}\lemlabel{interval-representation}
  For every $m$-vertex $t$-tree $H$, there exists a mapping $f$ assigning to every vertex $v$ in $H$ an interval $f(v)$ so that the following holds. Let $S:=\{f(v):v\in V(H)\}$. Then,
\begin{compactenum}
  \item $f(v)$ intersects $f(w)$, so $f(v)f(w)\in E(G_S)$, for every edge $vw\in E(H)$, and
  \item $\omega(G_S) \leq (t+1)\floor{\log_3(2m+1)+1}$.
\end{compactenum}
Furthermore, for every proper $(t+1)$-colouring $\varphi':V(H)\to\{1,\dots,t+1\}$ of $H$, there exists a colouring $\varphi'':V(H)\to\{1,\dots,\lfloor\log_3(2m+1)+1\rfloor\}$ such that $\phi:V(H)\to\N^2$ defined by $\phi(v):=(\varphi'(v),\varphi''(v))$ is a proper colouring of $G_S$.\footnote{This property is only used when discussing small optimizations in label lengths at the end of the paper.}
\end{lem}

In light of \lemref{interval-representation}, we call an \emph{interval representation} of $H$
a mapping $f$ assigning to every vertex $v$ of a graph $H$ an interval $f(v)$ in such a way that $f(v)$ and $f(w)$ intersect for every edge $vw\in E(H)$.
(Let us remark that $f(v)$ and $f(w)$ may or may not intersect when $v,w$ are two non-adjacent vertices, thus $H$ is a subgraph of the intersection graph of the intervals.)
We will always assume that all the endpoints of intervals in the representation are distinct.
This can be easily achieved by local perturbations not changing the intersection graph.

A finite set $X\subset\R$ \emph{stabs} a set $S$ of intervals if $X\cap [a,b]\neq\emptyset$ for every $[a,b]\in S$.
Let $H$ be a graph and $f$ be an interval representation of $H$.
We say that a binary search tree $T$ \emph{stabs} $H$ if $V(T)$ stabs the set of intervals $\{f(v):v\in V(H)\}$.
For $v$ a vertex in $H$, we let $x_T(v)$ denote the lowest common $T$-ancestor of $V(T)\cap f(v)$, see \figref{x}.
For $U \subseteq V(H)$, we let $x_T(U):=\{x_T(v):v\in U\}$.

\begin{figure}
  \begin{center}
    \includegraphics{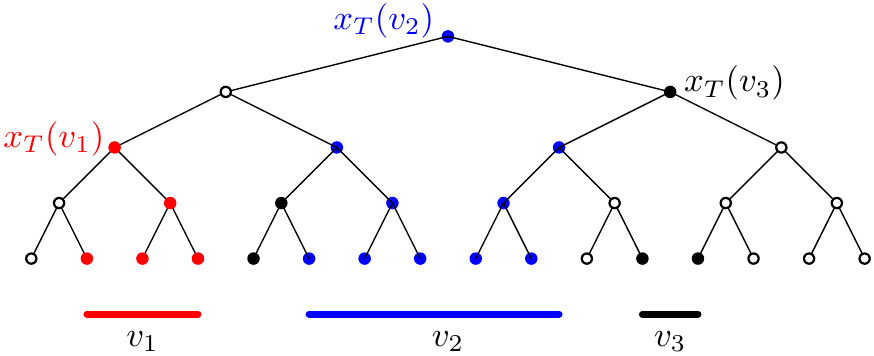}
  \end{center}
  \caption{The definition of $x_T(v)$.}
  \figlabel{x}
\end{figure}

\begin{lem}\lemlabel{all-the-obs-about-bst-and-intervals}
Let $H$ be a graph with a fixed interval representation $v\mapsto[a_v,b_v]$.
Let $T$ be a binary search tree that stabs $H$.
Then,
\begin{compactenum}
  \item for every vertex $v$ in $H$, we have $x_{T}(v) \in [a_v,b_v]$; and
  \label{itm:v-in-its-interval}
  \item for every clique $C$ in $H$, the set of nodes $x_T(C)$ lie in a single root-to-leaf path in $T$.
  \label{itm:clique-on-a-single-path}
\end{compactenum}
\end{lem}
\begin{proof}
  For the proof of the first item, consider $x:=x_T(v)$.
  Either we have $x\in [a_v,b_v]$, in which case there is nothing to prove,
  or there are two nodes $x_1,x_2\in V(T)\cap[a_v,b_v]$ such that $x_1$ is in the subtree of $T$ rooted at the left child of $x$ and $x_2$ is in the subtree of $T$ rooted at the right child of $x$.
  By the binary search tree property, $x_1<x<x_2$.
  But since $x_1,x_2 \in [a_v,b_v]$, we have $a_v\le x_1<x<x_2\le b_v$, so $x\in [a_v,b_v]$, as desired.

  For the proof of the second item, we just show it for $C$ being of size $2$, so for a single edge.
  The statement for general cliques will follow immediately by induction.
  Thus, consider two adjacent vertices $u_1$ and $u_2$ in $H$ and
  let $x_1=x_T(u_1)$ and $x_2=x_T(u_2)$.
  In order to get a contradiction, suppose that $x_1$ and $x_2$ do not lie on a single root-to-leaf path in $T$.
  Then there exists $x$ in $T$ such that $x_1$ is in the left subtree of $x$ and $x_2$ is in the right subtree of $x$.
  In particular, $x_1 < x < x_2$ by the binary search tree property.
  Since $u_1u_2 \in E(H)$ the corresponding intervals $[a_{u_1},b_{u_1}]$, $[a_{u_2}, b_{u_2}]$ intersect, so their union is an interval as well.
  Since $x_1,x_2$ lie in the union and $x_1 < x < x_2$, the node $x$ lies in the union as well.
  Therefore $x\in[a_{u_1},b_{u_1}]$ or $x\in [a_{u_2}, b_{u_2}]$.
  This contradicts the choice of $x_T(u_1)$ or $x_T(u_2)$ and completes the proof of the second item.
\end{proof}

Note that \lemref{all-the-obs-about-bst-and-intervals}(\ref{itm:v-in-its-interval}) implies that $x_T(v)$ has several alternative definitions:

\begin{cor}\corlabel{alternative-x-definitions}
    Let $H$ be a graph with a fixed interval representation $v\mapsto[a_v,b_v]$. Let $T$ be a binary search tree that stabs $H$ and let $v$ be any vertex of $H$.
    \begin{compactenum}
        \item The node set $V(T)\cap [a_v,b_v]$ has a unique node $x_T(v)$ of minimum $T$-depth.\label{itm:x-min-depth}.
        \item $T$ has a unique node $z=x_T(v)$ with the property that $z$ is the only node of $P_T(z)$ contained in $[a_v,b_v]$.\label{itm:x-unique-path}.
    \end{compactenum}
\end{cor}

\begin{proof}
    By definition $x_T(v)$ is the lowest common ancestor of $V(T)\cap [a_v,b_v]$ and by \lemref{all-the-obs-about-bst-and-intervals}(\ref{itm:v-in-its-interval}), $x_T(v)\in V(T)\cap [a_v,b_v]$ so $x_T(v)$ is the unique node of $V(T)\cap [a_v,b_v]$ of minimum $T$-depth. This establishes  the first item.

    For the second item, observe that the first item implies that $x_T(v)$ has the property that $x_T(v)$ is the only node of $P_T(x_T(v))$ contained in $[a_v,b_v]$.  To see that $x_T(v)$ is the unique node with this property consider any $z\in V(T)\cap[a_v,b_v]\setminus\{x_T(v)\}$. By (the original) definition of $x_T(v)$, $z$ is in the subtree of $T$ rooted at $x_T(v)$.  Since $z\neq x_T(v)$, $P_T(z)$ contains $x_T(v)$ as an internal node so $z$ is not the only node of $P_T(z)$ contained in $[a_v,b_v]$.
\end{proof}

\subsection{A Labelling Scheme for $t$-Trees}
\label{sec:t-trees}

We describe a labelling scheme for $t$-trees that, like our labelling scheme for paths, is based on a binary search tree.  The ideas behind this scheme are not new; this is essentially the labelling scheme for $t$-trees described by Gavoille and Labourel \cite{gavoille.labourel:shorter}.  However, we present these ideas in a manner that makes it natural to generalize the results of \secref{pxp}.

We are given a $t$-tree $H$ on $m$ vertices with an interval representation $v\mapsto[a_v,b_v]$ as in~\lemref{interval-representation}.  In particular, the clique number of the resulting interval graph is at most $(t+1)\floor{\log_3(2m+1)+1}$.  Since interval graphs are perfect, their clique number coincides with their chromatic number.  Let $\varphi: V(H) \to [(t+1)\floor{\log_3(2m+1)+1}]$ be a colouring such that $u$ and $v$ have distinct colours whenever the intervals of $u$ and $v$ intersect.

The following easy observation shows that a vertex $v$ of $H$ is uniquely identified by $\varphi(v)$ and the $x_T(v)$ value in a binary search tree $T$ that stabs $H$.
This gives ground for an adjacency labelling scheme.
\begin{obs}\obslabel{unique-id}
    Let $T$ be a binary search tree that stabs $H$.
    Let $v$ and $w$ be two distinct vertices in $H$.
    Then, $x_T(v)\neq x_T(w)$ or $\varphi(v)\neq\varphi(w)$.
    Consequently, $\sigma_T(x_T(v))\neq \sigma_T(x_T(w))$ or $\varphi(v)\neq\varphi(w)$.
\end{obs}

\begin{proof}
  If $x_T(v)=x=x_T(w)$, then by \lemref{all-the-obs-about-bst-and-intervals}\itemref{itm:v-in-its-interval} intervals $[a_v,b_v]$ and $[a_w,b_w]$ each contain $x$,
  so they intersect.
  Therefore, $\varphi(v)\neq\varphi(w)$.
\end{proof}

Let $T$ be a binary search tree that stabs $H$ and let $v$ be a vertex in $H$.
Fix an elimination ordering of $H$.
Recall that by \lemref{all-the-obs-about-bst-and-intervals}\itemref{itm:clique-on-a-single-path},
for every vertex $v$ in $H$,
all the nodes in $x_T(C_H(v))$ lie on a single root-to-leaf path in $T$.
We define
\[
\sigma_{H,T}(v):=\sigma_T(x),\quad\text{where $x$ is the node in $x_T(C_H(v))$ of maximum depth in $T$.}
\]
Note that $d_T(x_T(v))\le |\sigma_{H,T}(v)|$, and equality holds only if $x_T(v)$ is the deepest node in $x_T(C_H(v))$.

Now, we define the label $\tau_{H,T}(v)$ of a vertex $v$ in $H$.
Let $d=|C_H(v)|$ and let $u_1,\dots, u_d$ be the vertices in $C_H(v)$, in any order, so $v$ is one of them.
Recall that $\gamma$ is the Elias encoding of natural numbers.
The label $\tau_{H,T}(v)$ is defined
as the concatenation of
\begin{compactenum}[(T1)]
  \item $\gamma(|\sigma_{H,T}(v)|)$ and $\sigma_{H,T}(v)$;
  \item $\gamma(\varphi(v))$;
  \item $\gamma(d)$;
  \item $\gamma(d_{T}(x_T(u_i)))$ for each $i\in\{1,\dots,d\}$;
  \item $\gamma(\varphi(u_i))$ for each $i\in\{1,\dots,d\}$.
  \end{compactenum}

\begin{lem}\lemlabel{t-tree-labelling}
  There exists a function $F:(\{0,1\}^*)^2\to\{0,-1,1,\perp\}$ such that
  for any $t$-tree $H$
  with a fixed elimination ordering and a fixed interval representation,
  and a proper colouring $\varphi$ of the interval representation, and
  for any binary search tree $T$ stabbing $H$,
  for any two vertices $v$, $w$ in $H$, we have
  \[
      F(\tau_{H,T}(v),\tau_{H,T}(w)) = \begin{cases}
      0 & \text{if $v=w$;} \\
      -1 & \text{if $v$ and $w$ are adjacent in $H$, and $w\in C_H(v)$;} \\
      1 & \text{if $v$ and $w$ are adjacent in $H$, and $v\in C_H(w)$;} \\
      \perp & \text{otherwise.}
    \end{cases}
  \]
\end{lem}
Note that the labels $\tau_{H,T}(v)$ depend on the choice of elimination ordering and interval representation of $H$ but the function $F$ does not.
Recall also that $d_T(x_T(u)) \le |\sigma_{H,T}(u)|\leq h(T)$, for all vertices $u$ in $H$.
Moreover, $|C_H(u)|\leq t+1$ for all vertices $u$ in $H$.
Thus when $H$ is an $m$-vertex $t$-tree and $\varphi$ takes values bounded in $\Oh(t\log m)$,
we get labels $\tau_{H,T}(v)$ of length
$h(T) + \Oh(t\cdot (\log h(T) + \log t + \log\log m))$.
In particular, we can take a perfectly balanced tree $T$ whose vertex set $V(T)$ is just the set of all endpoints of intervals representing $H$.
Then $T$ stabs $H$ and $h(T)=\lfloor\log(2m)\rfloor$.
This way the labels $\tau_{H,T}(v)$ are of length $\log m + \Oh(t\log\log m +t\log t)$.

\begin{proof}[Proof of~\lemref{t-tree-labelling}]
For the adjacency testing, first note that from a given label $\tau_{H,T}(v)$ of a vertex $v$ in $H$, we can decode each block of the label.
This is just because $\gamma$, the Elias encoding, is prefix-free.

Note that from the blocks of $\tau_{H,T}(v)$ we can determine $\sigma_T(x_T(v))$, $\varphi(v)$, and
$\sigma_T(x_T(u))$, $\varphi(u)$, for all $u\in C_H(v)$.

Given the labels of two vertices $v$ and $w$ in $H$ we can test if they are adjacent as follows.
\begin{enumerate}
  \item If $\sigma_T(x_T(v)) = \sigma_T(x_T(w))$ and $\varphi(v)=\varphi(w)$, we conclude that $v=w$ (by~\obsref{unique-id}), so $F$ outputs $0$ in this case.

  \item Let $d=|C_H(v)|$ and $u_1,\dots, u_{d}$ be the vertices in $C_H(v)$.
  From the label of $v$,
  we decode the values of $d$, $\sigma_T(x_T(u_i))$ and $\varphi(u_i)$ for each $i\in\{1,\dots,d\}$.
  If $\sigma_T(x_T(w)) = \sigma_T(x_T(u_i))$ and $\varphi(w)=\varphi(u_i)$ for some $i\in\{1,\dots,d\}$, then we conclude that $w=u_i$ (by~\obsref{unique-id}) and $F$ outputs $-1$.

  \item  Now, let $d=|C_H(w)|$ and let $u_1,\dots, u_{d}$ be the vertices in $C_H(w)$.
  From the label of $w$,
  we decode the values of $d$, $\sigma_T(x_T(u_i))$ and $\varphi(u_i)$ for each $i\in\{1,\dots,d\}$.
  If $\sigma_T(x_T(v)) = \sigma_T(x_T(u_i))$ and $\varphi(v)=\varphi(u_i)$ for some $i\in\{1,\dots,d\}$, then we conclude that $v=u_i$ (by~\obsref{unique-id}) and $F$ outputs $1$.

  \item Otherwise, $v\neq w$, $v\not\in C_{H}(w)$, and $w\not\in C_{H}(v)$.
  This implies that $v$ and $w$ are not adjacent in $H$ because each edge in $H$ connects a vertex $u$ with a vertex in $C_H(u)$, for some $u$ in $H$.
  Thus, in this case $F$ outputs $\perp$.
\end{enumerate}
\end{proof}


In fact, we can get labels of length $\log{m} + \Oh(t\log\log{m})$ instead of $\log{m} + \Oh(t\log\log{m}+ t\log t)$.
This can be done by improving the length of (T5) from $\Oh ( t \log( t  \log m ))$ to $\Oh ( t  \log\log m )$, as we now explain.
In order to achieve this we need to work with the colouring $(\varphi'(v),\varphi''(v))$ of $H$ given by Lemma~\lemref{interval-representation} instead of $\varphi(v)$.
Recall that the image of $\varphi'$ is $\{1,\ldots,t+1\}$ and the image of $\varphi''$ is $\{1,\ldots,\floor{\log_3(2m+1)+1}\}$.
Given a vertex $v$ of $H$, we order the vertices $u_1, \dots, u_d$ in $C_H(v)$ according to their $\varphi'$-colours.
Now the improved (T5) block of the labelling is an array $R$ of $t+1$ entries indexed by $\varphi'$ colours.
We set $R[\varphi'(w)] := \varphi''(w)$, for each $w\in C_H(v)$.
Note that this may leave some entries undefined if $d<t+1$, in which case we set these to $0$ to distinguish them from ``true'' colours $\{1,\ldots,t+1\}$.
This way only $\Oh(t\log\log{m})$ bits suffice to encode $R$.

\subsection{Interval Transition Labels}

We now show that the solution presented in \secref{pxp} generalizes to the current setting.

Let $G$ be an $n$-vertex subgraph of $H\boxtimes P$ where $H$ is an $m$-vertex $t$-tree and $P=1,\dots,h$ is a path. Clearly, we can assume that $m\leq n$ and $h\leq n$.

Fix an elimination ordering of $H$ and an interval representation of $H$ with clique number at most $(t+1)\floor{\log_3(2m+1)+1}$, see~\lemref{interval-representation}.  Let $\varphi:V(H)\to\{1,\dots, (t+1)\floor{\log_3(2m+1)+1} \}$ be a proper colouring of the interval representation of $H$.

For each $y\in\{1,\dots,h\}$, let
\begin{align*}
S_y&=\{v\in V(H): (v,y)\in V(G)\}, \text{ and}\\
S^+_y&=\textstyle\bigcup_{v\in S_y} C_H(v).
\end{align*}
Note that $\sum_{y=1}^h |S_y| =n$ and $\sum_{y=1}^h |S^+_y| \leq (t+1)n$.
Let $S_0=\emptyset$.
For each $y\in\{1,\dots,h\}$, let $X_y\subset\R$ be the set of all endpoints of intervals representing vertices in $S^+_y\cup S^+_{y-1}$.
Apply \lemref{fractional} to the sequence $X_1,\dots,X_h$ to obtain a $1$-chunking sequence $V_1,\dots,V_{h}$ such that $V_y\supseteq X_y$ for each $y\in\{1,\dots,h\}$, and $\sum_{y=1}^h |V_y|\le 4\sum_{y=1}^h |X_y|$. Let $T_1,\dots,T_h$ be a $k$-bulk tree sequence based on $V_1,\dots,V_{h}$ with $k := \max\left\{5,\left\lceil\sqrt{\log n / \log\log n}\right\rceil\right\}$ (recall that if we fix the starting perfectly balanced binary search tree $T_1$ with vertex set $V_1$, this sequence exists and is unique).

For each $y\in\{1,\dots,h\}$, let $H^+_y$ be the subgraph of $H$ induced by $S^+_y\cup S^+_{y-1}$.
In particular, each $H^+_y$ is a $t$-tree.
We fix the elimination ordering of $H^+_y$ inherited from $H$.
Similarly, we fix the interval representation of $H^+_y$ and the proper colouring $\varphi$ of the interval representation, as the projection of respective ones for $H$.

Since $X_y \subseteq V_y = V(T_y)$, we have that $T_y$ stabs $H^+_y$.

By construction, we have
\[
\textstyle\sum_{y=1}^h |T_y| = \textstyle\sum_{y=1}^h |V_y| \le 4\textstyle\sum_{y=1}^h |X_y| \le 8\textstyle\sum_{y=1}^h |S^+_y\cup S^+_{y-1}| \le 16(t+1)n.
\]

By \lemref{bulk-tree-height}, for each $y\in\{1,\dots,h\}$ we have
\begin{align*}
h(T_y)&= \log |T_y| + \Oh(k+k^{-1}\log |T_y|)\\
&\leq \log |T_y| + \Oh(k+k^{-1}\log n).
\end{align*}

The following lemma, which is analogous to \lemref{node-transitions}, is the last piece of the puzzle needed for an adjacency labelling scheme for subgraphs of $H\boxtimes P$.

\begin{lem}\lemlabel{interval-transitions}
  There exists a function $J:(\{0,1\}^*)^2\to \{0,1\}^*$ such that,
  for any $H$, $P$, $G$, $\varphi$, $S_1,\dots,S_h$, $X_1,\dots,X_h$, and each $k$-bulk tree sequence $T_1,\dots,T_h$ defined as above, for each $y\in\{1,\dots,h-1\}$ and each $v\in S_y$, there exists $\mu_y(v)\in\{0,1\}^*$ with $|\mu_y(v)|= \Oh(k\log h(T_y))$ such that $J(\sigma_{H^+_y,T_y}(v), \mu_y(v))=\sigma_{H^+_{y+1},T_{y+1}}(v)$.
\end{lem}

For strings $a$ and $b$, let $a\preceq b$ denote that $a$ is a prefix of $b$ and let $a\prec b$ denote that $a\preceq b$ and $a\neq b$.

\begin{proof}
  Let $v \in S_y$.
  First of all, note that $v$ is a vertex in $H^+_y$ and $H^+_{y+1}$.
  Since $T_y$ stabs $H^+_y$ and $T_{y+1}$ stabs $H^+_{y+1}$,
  we conclude that $\sigma_{H^+_{y},T_y}(v)$, $\sigma_{H^+_{y+1},T_{y+1}}(v)$ are well-defined.

  As in the proof of \lemref{node-bal-transitions} and \lemref{node-transitions}, we must dig into the details of the three bulk tree operations that transform $T_y$ into $T_{y+1}$.
  Let $I:=V(T_{y+1})\setminus V(T_{y})$ and $D:=V(T_y)\setminus V(T_{y+1})$.
  Recall that the three steps are:
  applying $\textsc{BulkBalance}(\theta,k)$ to $T_y$ with the appropriate value of $\theta$ to obtain $T'$,
  applying $\textsc{BulkInsert}(I)$ to $T'$ to obtain $T''$, and
  applying $\textsc{BulkDelete}(D)$ to $T''$ to obtain $T_{y+1}$.
  Recall that whenever $\textsc{BulkBalance}(\theta,k)$ is applied in this context we have $\theta < h(T_y)$.

  By \lemref{deletion-signature}, \lemref{insertion-depth}, and \lemref{balance-depth} we have  $h(T_{y+1}) \leq h(T'') \leq h(T')+2 \leq h(T_y)+3$.
  Thus the heights of all these trees are $\Oh(h(T_y))$.

  The transition code $\mu_y(v)$ is the concatenation of two parts $\mu_y^{\textsc{Bal}}(v)$ and $\mu_y^{\textsc{Del}}(v)$ devoted to different steps of the transformation from $T_y$ to $T_{y+1}$.
  First, the code $\mu_y^{\textsc{Bal}}(v)$ will serve to move from $\sigma_{H^+_y,T_y}(v)$ to $\sigma_{H^+_y,T'}(v)$.
  Next, we will argue that $\sigma_{H^+_y,T'}(v) = \sigma_{H^+_{y+1},T''}(v)$.
  Then the code $\mu_y^{\textsc{Del}}(v)$ will serve to move from $\sigma_{H^+_{y+1},T''}(v)$ to $\sigma_{H^+_{y+1},T_{y+1}}(v)$.

  We start with a discussion on rebalancing that leads to the definition of $\mu_y^{\textsc{Bal}}(v)$.
  The tree $T_y$ is rebalanced by an application of \textsc{BulkBalance}$(\theta,k)$ with an appropriate value of $\theta < h(T_y)$ and the resulting tree is $T'$.
  Recall that \textsc{BulkBalance}$(\theta,k)$ calls $\textsc{Balance}(x,k)$ for each node $x$ of depth-$\theta$ in $T_y$.
  Recall also that the changes made by $\textsc{Balance}(x,k)$ are limited to the subtree of $T_y$ rooted at $x$.
  See~\figref{rebalance-t-tree}.
  Let $Q$ be the path in $T_y$ encoded by $\sigma_{H^+_{y},T_{y}}(v)$.
  Thus $Q$ is the path from the root of $T_y$ to the deepest node $z$ in $x_{T_y}(C_{H^+_{y}}(v))$.
  Clearly, if $Q$ is not hitting vertices of depth at least $\theta$, then $\sigma_{H^+_{y},T_{y}}(v)=\sigma_{H^+_{y},T'}(v)$.
  Thus, in the case that $|\sigma_{H^+_{y},T_{y}}(v)| < \theta$, we define
  \[
  \mu_y^{\textsc{Bal}}(v) := \gamma(\theta).
  \]
  Note that in this case $|\mu_y^{\textsc{Bal}}(v)| = \Oh(\log \theta) = \Oh(\log h(T_y))$.

  Assume now that $|\sigma_{H^+_{y},T_{y}}(v)| \geq \theta$ and let $x$ be the $T_y$-ancestor of $z$ at depth $\theta$.
  Let $T_*$ be the subtree of $T_y$ rooted at $x$ and let $T_*'$ be the new tree obtained after calling $\textsc{Balance}(x,k)$ on the root, $x$, of $T_*$. (So $T_*$ is a subtree of $T_y$ and $T_*'$ is a subtree of $T'$.)
  Recall that the application of $\textsc{Balance}(x,k)$ identifies two sets of nodes $Z$ and $X$ that eventually form a perfectly balanced tree $\hat{T_0}$ of height at most $k$ which forms the top part of $T_*'$.
  Let $Q'$ be the path in $T'$ encoded by $\sigma_{H^+_y,T'}(v)$, so $Q'$ is the path from the root of $T'$ to the deepest node $z'$ in $x_{T'}(C_{H^+_y}(v))$.

  If $z' \in Z\cup X$, then $\sigma_{H^+_y,T'}(v) = \sigma_{T_y}(x),\sigma_{\hat{T_0}}(z')$.
  In this case, we define
  \[
  \mu_y^{\textsc{Bal}}(v) := \gamma(\theta),0,\gamma(|\sigma_{\hat{T_0}}(z')|),\sigma_{\hat{T_0}}(z').
  \]
  Note that in this case $|\mu_y^{\textsc{Bal}}(v)|=\Oh(\log\theta)+\Oh(\log k) + \Oh(k) = \Oh(\log h(T_y)+k)$.

  \begin{figure}
    \begin{center}
      \includegraphics{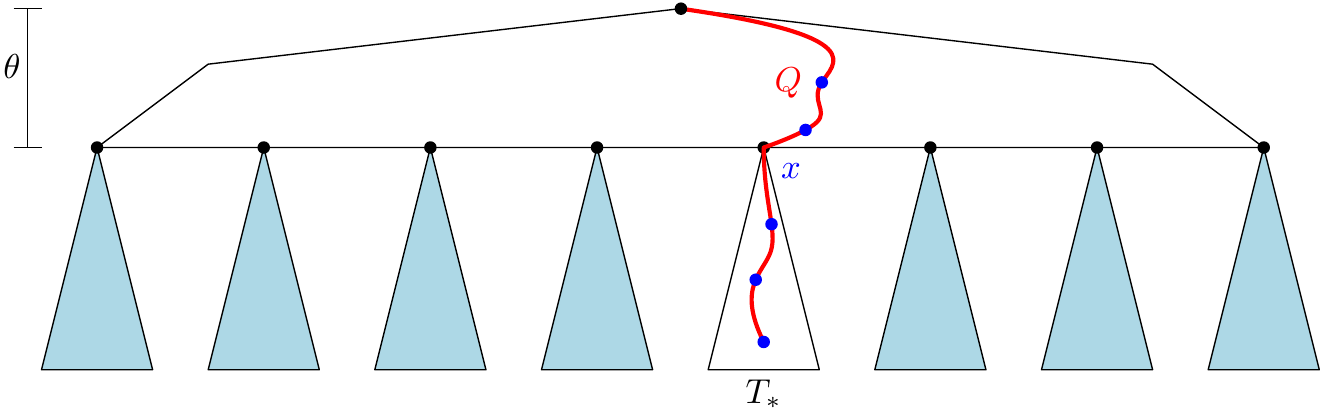}
    \end{center}
    \caption{Only a call to $\textsc{Balance}(x,k)$ on a node $x$ lying on a path described by $\sigma_{H^+_y,T_y}(v)$ can affect $\sigma_{H^+_y,T'}(v)$.}
    \figlabel{rebalance-t-tree}
  \end{figure}

  Now we are left with the case $|\sigma_{H^+_{y},T_{y}}(v)| \geq \theta$ and $z'\not\in Z\cup X$.
  Let $w \in C_{H^+_y}(v)$ be a vertex witnessing $z' = x_{T'}(w)$.
  By \corref{alternative-x-definitions}(\ref{itm:x-unique-path}),
  $z'$ is the unique node of $Q'$ contained in $[a_w,b_w]$, i.e., the interval representing $w$.
  By the properties of a binary search tree, we conclude that
  \[
  [a_w,b_w]\cap (Z\cup X)=\emptyset.
  \]
  Consider now the node $x_{T_y}(w)$ in $T_y$. Since $[a_w,b_w]\cap Z=\emptyset$ we know that $x_{T_y}(w)$ lie in one of the trees, say $T_i$, of the forest $T_*-Z$.
  Recall that $\textsc{Balance}(x,k)$ calls $\textsc{MultiSplit}(x_{i,1},\dots,x_{i,c_i})$ on $T_{i}$ to obtain a sequence of trees $T_{i,0},\dots,T_{i,c_i}$.
  This, in turn results in zero or more calls to $\textsc{Split}(x')$ for nodes $x'\in V(T_{i})\cap X$. The following claim explains the effect of one individual call to $\textsc{Split}(x')$:

  \begin{clm}\clmlabel{x-switch}
    Let $T$ be a binary search tree that stabs an interval $[a_w,b_w]$,
    and let $T_{<x}$ and $T_{>x}$ be the two trees resulting from calling $\textsc{Split}(x)$ on $T$ for some $x\in V(T)$. Then, exactly one of the following is true:
    \begin{compactenum}
      \item $a_w\le x\le b_w$;
      \item $x< a_w$, in which case $x_{T_{>x}}(w)=x_T(w)$; or
      \item $b_w < x$, in which case $x_{T_{<x}}(w)=x_T(w)$.
    \end{compactenum}
  \end{clm}
  \begin{proof}
    That the first condition of exactly one of the three cases applies is obvious.  Case~(1) has no specific requirements and Cases~(2) and (3) are symmetric, so we focus on Case~(2), so $x < a_w\le b_w$.

    By \corref{alternative-x-definitions}(\ref{itm:x-unique-path}), $z=x_T(v)$ is the unique node $z$ in $T$ such that the path $P$ from the root to $z$ has exactly one node in $[a_w,b_w]$.
    Recall that by construction
    the path $P'$ from the root to $z$ in $T_{>x}$ is obtained from $P$
    by deleting all values less than or equal to $x$.
    Therefore the path $P'$ in $T_{>x}$ still has exactly one node in $[a_w,b_w]$, namely $z$,
    so $z=x_{T_{>x}}(w)$. This completes the proof of \clmref{x-switch}.
  \end{proof}

  Since all the calls to $\textsc{Split}(x')$ generated by $\textsc{MultiSplit}(x_{i,1},\dots,x_{i,c_i})$ on the subtree $T_{i}$ are called with $x'\in X$ and $[a_w,b_w]\cap X =\emptyset$, \clmref{x-switch} guarantees that
  \begin{align*}
  x_{T_y}(w) = x_{T'}(w) &= z', \textrm{ and}\\
  \sigma_{H^+_y,T'}(v) &= \sigma_{T'}(z').
  \end{align*}

  Finally, by \lemref{node-bal-transitions} there exists a function $B$ and a bitstring
  $\nu(z')$ of length $\Oh(k\log h(T_y))$ such that $B(\sigma_{T_y}(z'), \nu(z')) = \sigma_{T'}(z')$.

  All this justifies the following definition, in the case that $|\sigma_{H^+_{y},T_{y}}(v)| \geq \theta$ and $z'\not\in Z\cup X$:
  \[
   \mu_y^{\textsc{Bal}}(v) := \gamma(\theta),1,\gamma(|\sigma_{T_y}(x_{T_y}(w))|),\gamma(|\nu(x_{T_y}(w))|),\nu(x_{T_y}(w)).
  \]
  Note that in this case $|\mu_y^{\textsc{Bal}}(v)|=\Oh(\log\theta)+\Oh(\log h(T_y)) + \Oh(k\log h(T_y)) = \Oh(k\log h(T_y))$.


  Now we shall argue that
  \[
  \sigma_{H^+_{y},T'}(v) = \sigma_{H^+_{y},T''}(v) = \sigma_{H^+_{y+1},T''}(v).
  \]
  Recall that $T''$ comes as a result of an application of $\textsc{BulkInsert}(I)$ to $T'$
  that attaches some small subtrees to the leaves of $T'$.
  This way, for every $u \in C_{H_y^+}(v) = C_H(v) \subseteq S^+_y \subseteq V(H^+_y)$,
  we have that $x_{T'}(u)$ is a $T''$-ancestor of any node $x$ in $T''$ such that
  $x\in I$ and $x\in [a_u,b_u]$.
  Hence, $x_{T'}(u) = x_{T''}(u)$ and
  $\sigma_{T'}(x_{T'}(u)) = \sigma_{T''}(x_{T''}(u))$.
  Therefore,
  $\sigma_{H^+_{y},T'}(v) =  \sigma_{H^+_{y},T''}(v)$.
  Recall that $\sigma_{H^+_{y},T''}(v)$ and $\sigma_{H^+_{y+1},T''}(v)$ encode paths in $T''$ from the root to the deepest node in $x_{T''}(C_{H^+_{y}}(v))$ and in $x_{T''}(C_{H^+_{y+1}}(v))$, respectively.
  Again, since $v\in S_y$ we have that $C_H(v)\subseteq S^+_y \subseteq V(H^+_y) \cap V(H^+_{y+1})$. Therefore,
  $C_{H^+_{y}}(v) = C_H(v) = C_{H^+_{y+1}}(v)$ and $\sigma_{H^+_{y},T''}(v) =\sigma_{H^+_{y+1},T''}(v)$.

  Next we describe $\mu_y^{\textsc{Del}}(v)$.
  The transition code $\mu_y^{\textsc{Del}}(v)$ serves to move
  from $\sigma_{H^+_{y+1},T''}(v)$ to $\sigma_{H^+_{y+1},T_{y+1}}(v)$.
  An application of $\textsc{BulkDelete}(D)$ to $T''$ results in
  a sequence of individual deletions.
  Consider a single deletion of an element $x$ and let $T^{\text{bef}}$ and $T^{\text{aft}}$ denote the trees before and after the deletion, respectively.

  \begin{clm}\clmlabel{deletion-ancestor}
  For every $u\in S^+_y$, we have
  \[
  \sigma_{T^{\text{aft}}}(x_{T^{\text{aft}}}(u))\preceq\sigma_{T^{\text{bef}}}(x_{T^{\text{bef}}}(u)).
  \]
  \end{clm}

  \begin{proof}[Proof of \clmref{deletion-ancestor}]
    See \figref{deletion-t-tree}.  At a global level, the deletion of a value $x$ from $T^{\text{bef}}$ involves finding a sequence of consecutive values $x_0<x_1<\cdots<x_r$ or $x_0>x_1>\cdots>x_r$ where $x=x_0$, $x_r$ is a leaf and $x_{i-1}$ is a $T^{\text{bef}}$-ancestor of $x_{i}$ for each $i\in\{1,\dots,r\}$.  The leaf containing $x_r$ is deleted and, for each $i\in\{0,\dots,r-1\}$, the (value of) node $x_i$ is replaced with (the value in) node $x_{i+1}$.
    The resulting tree is $T^{\text{aft}}$.
    \begin{figure}
      \begin{center}
        \includegraphics{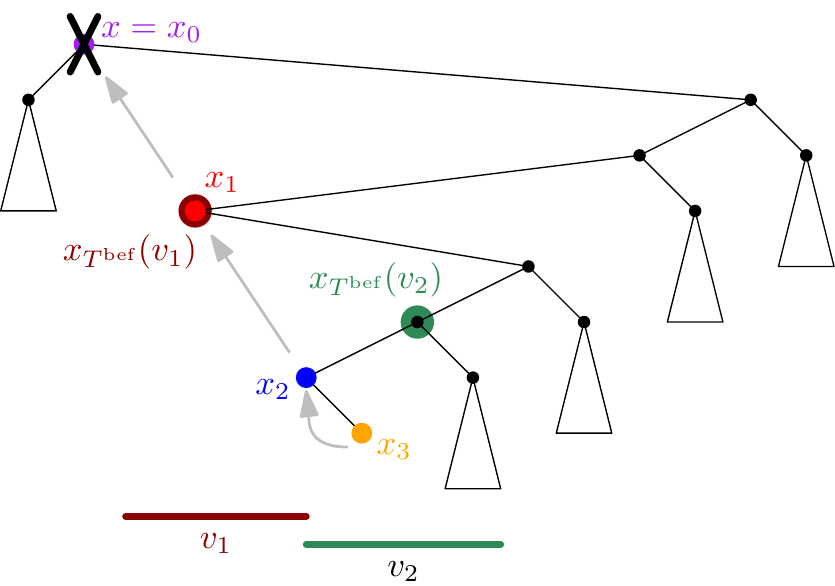}\\[1ex]
        $\Downarrow$\\[1ex]
        \includegraphics{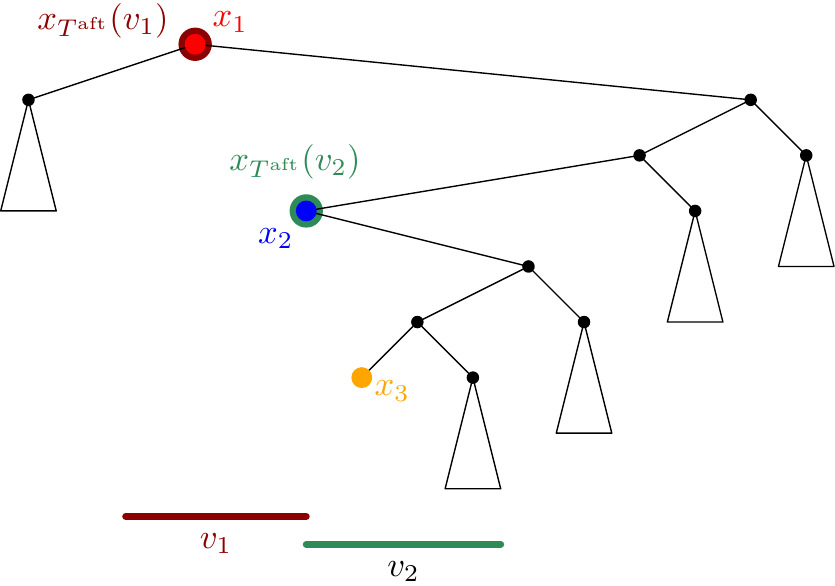}
      \end{center}
      \caption{The effect of a single deletion on $x_{T^{\text{bef}}}(v)$.}
      \figlabel{deletion-t-tree}
    \end{figure}

    First we look at the case $x_{T^{\text{bef}}}(u)=x_0$.
    Recall that $u\in S^+_y$, so $X_{y+1}$ contains both endpoints of the interval representing $u$, say $[a_u,b_u]$.
    This means that $x_0\neq a_u$ and $x_0\neq b_u$ and both endpoints lie in the subtree of $T^{\text{bef}}$ rooted at $x_0$.
    In particular, $x_0$ is not a leaf and $r\geq 1$.
    By~\lemref{all-the-obs-about-bst-and-intervals}\itemref{itm:v-in-its-interval}, we have $a_u < x_0 < b_u$.
    Since $x_1$ is the smallest value in the right subtree of $x_0$ or the largest value in the left subtree of $x_0$, we conclude that
    $a_u\leq x_1\leq b_u$.
    Thus, in this case we have $x_{T^{\text{aft}}}(u)=x_1$ and
    \[
    \sigma_{T^{\text{aft}}}(x_{T^{\text{aft}}}(u)) = \sigma_{T^{\text{aft}}}(x_1) = \sigma_{T^{\text{bef}}}(x_0) = \sigma_{T^{\text{bef}}}(x_{T^{\text{bef}}}(u)).
    \]

    If $x_{T^{\text{bef}}}(u)=x_i$ for some $i\in\{1,\dots,r\}$,
    then $x_{T^{\text{aft}}}(u)=x_i$ (see the interval $v_1$ in \figref{deletion-t-tree}) and
    \[
    \sigma_{T^{\text{aft}}}(x_{T^{\text{aft}}}(u))= \sigma_{T^{\text{aft}}}(x_i) = \sigma_{T^{\text{bef}}}(x_{i-1}) \prec \sigma_{T^{\text{bef}}}(x_{i}) = \sigma_{T^{\text{bef}}}(x_{T^{\text{bef}}}(u)).\]

    Finally, if $\sigma_{T^{\text{bef}}}(x_{T^{\text{bef}}}(u))\neq \sigma_{T^{\text{aft}}}(x_{T^{\text{aft}}}(u))$ and
    $x_{T^{\text{bef}}}(u)\neq x_i$ for all $i\in\{0,\dots,r\}$,
    then the only possibility is that $x_{T^{\text{aft}}}(u)=x_i$ for some $x_i\in [a_u,b_u]$ (see the interval $v_2$ in \figref{deletion-t-tree}).
    This can only happen if $x_{T^{\text{bef}}}(u)$ is a $T^{\text{bef}}$-ancestor of $x_i$ and $x_i$ is a $T^{\text{aft}}$-ancestor of $x_{T^{\text{bef}}}(u)$.
    The latter is equivalent with the fact that $x_{i-1}$ is a $T^{\text{bef}}$-ancestor of $x_{T^{\text{bef}}}(u)$.
    Therefore, we have
    \[
    \sigma_{T^{\text{aft}}}(x_{T^{\text{aft}}}(u)) = \sigma_{T^{\text{aft}}}(x_i) =
    \sigma_{T^{\text{bef}}}(x_{i-1}) \prec \sigma_{T^{\text{bef}}}(x_{T^{\text{bef}}}(u)).
    \]
    This completes the proof of the claim.
  \end{proof}

Since $\textsc{BulkDelete}(D)$ is a sequence of individual deletions, by multiple applications of~\clmref{deletion-ancestor} we get that
  \begin{align*}
    \sigma_{H^+_{y+1},T_{y+1}}(v)
      & = \sigma_{T_{y+1}}(x_{T_{y+1}}(u)) & \text{(for some $u\in C_{H^+_{y+1}}(v)=C_H(v)$)} \\
      & \preceq \sigma_{T''}(x_{T''}(u)) & \text{(by \clmref{deletion-ancestor})} \\
      & \preceq \sigma_{H^+_{y+1},T''}(v). 
  \end{align*}

We define
  \[\mu_y^{\textsc{Del}}(v) := \gamma(|\sigma_{H^{+}_{y+1},T_{y+1}}(v)|).
  \]
Note that $|\mu_y^{\textsc{Del}}(v)| = \Oh(\log h(T_{y+1})) = \Oh(\log h(T_y))$.

  The function $J$ is defined as expected:
  Given $\sigma_{H^+_y,T_{y}}(v)$ and $\mu_y(v)$,
  the function $J$ first decodes the value of $\theta$
  which is always the first block of $\mu_y^{\textsc{Bal}}(v)$.
  If $|\sigma_{H^+_y,T_{y}}(v)| < \theta$, then
  $J$ concludes that $\sigma_{H^+_y,T'}(v) = \sigma_{H^+_y,T_{y}}(v)$.
  Otherwise, in the case $|\sigma_{H^+_y,T_{y}}(v)| \geq \theta$,
  the function $J$ reads the next bit of $\mu_y(v)$.
  If it is $0$ then $J$ decodes the value of $\sigma_{\hat{T_0}}(z')$,
  computes $\sigma_{T_y}(x)$ which is the prefix of $\sigma_{H^+_y,T_{y}}(v)$ of length $\theta$,
  and concludes that $\sigma_{H^+_y,T'}(v) = \sigma_{T_y}(x),\sigma_{\hat{T_0}}(z')$.
  Otherwise, the bit is $1$.
  In this case, first $J$ decodes $|\sigma_{T_y}(x_{T_y}(w))|$ and $\nu(x_{T_y}(w))$.
  Next, $J$ computes $\sigma_{T_y}(x_{T_y}(w))$ which is the prefix of $\sigma_{H^+_y,T_{y}}(v)$ of length $|\sigma_{T_y}(x_{T_y}(w))|$.
  Next, $J$ computes $B(\sigma_{T_y}(x_{T_y}(w)),\nu(x_{T_y}(w))) = \sigma_{T'}(x_{T'}(w))$ and in this case
  $J$ concludes that $\sigma_{H^+_y,T'}(v) = \sigma_{T'}(x_{T'}(w))$.

  Thus, in either case $J$ establishes the value of $\sigma_{H^+_y,T'}(v) = \sigma_{H^+_{y+1},T''}(v)$.
  Now, $J$ looks up $\mu_y^{\textsc{Del}}(v)$ and decodes $|\sigma_{H^{+}_{y+1},T_{y+1}}(v)|$.
  The value of $\sigma_{H^{+}_{y+1},T_{y+1}}(v)$ is simply the prefix of $\sigma_{H^+_{y+1},T''}(v)$ of length $|\sigma_{H^{+}_{y+1},T_{y+1}}(v)|$.

  This completes the proof of the lemma.
\end{proof}

\subsection{The Labels}

We are ready to combine everything together and devise labels for vertices of $G$.

Let $A:(\{0,1\}^{*})^2\to\{0,1\}^*$ be the function given by \lemref{row-code} such that,
using the weight function $w(y):=|T_y|$ for each $y\in\{1,\dots,h\}$,
we have a prefix-free code $\alpha:\{1,\dots,h\}\to\{0,1\}^*$ such that
\begin{align*}
|\alpha(y)|&=\log\left(\textstyle\sum_{i=1}^h|T_i|\right) - \log|T_y| + \Oh(\log\log h)\\
&\leq \log\bigl( 16(t+1)n\bigr) - \log|T_y| + \Oh(\log\log n)\\
&\leq \log n -\log|T_y| + \Oh(\log\log n + \log t),
\end{align*}
for each $y\in\{1,\dots,h\}$, and $A(\alpha(i),\alpha(j))$ outputs $0$, $1$, $-1$, or $\perp$, depending whether the value of $j$ is $i$, $i+1$, $i-1$, or some other value, respectively.

Let $F:(\{0,1\}^{*})^2\to\{0,1\}^*$ be the function given by \lemref{t-tree-labelling}.

Let $J:(\{0,1\}^{*})^2\to\{0,1\}^*$ be the function given by \lemref{interval-transitions} such that
for each $y\in\{1,\dots,h-1\}$ and each $v\in S_{y}$,
there exists a code $\mu_{y}(v)$ with $|\mu_{y}(v)|=\Oh(k\log h(T_{y}))=\Oh(k\log\log n+k \log k)$ such that $J(\sigma_{H^+_y,T_{y}}(v),\mu_{y}(v))=\sigma_{H^+_{y+1},T_{y+1}}(v)$.

Let $z=(v,y)$ be a vertex in $G$.
Recall that $S^+_y \subseteq V(H^+_y)$ and if $i\neq h$ then also $S^+_y\subseteq V(H^+_{y+1})$.
Therefore
$C_{H}(v) = C_{H^+_y}(v) = C_{H^+_{y+1}}(v)$.
Let $d = |C_{H}(v)|$ and let $u_1,\dots,u_d$ be the vertices of $C_H(v)$.
Recall that $d\leq t+1$.
We define $a(z)$ to be an array of $3d$ bits indicating whether
each of the edges between $(v,y)$ and $(u_i,y+\{-1,0,1\})$ are present in $G$.
The label of $z=(v,y)$ is the concatenation of the following bitstrings:

\begin{compactenum}[(L1)]
  \item\label{label-alpha} $\alpha(y)$;
  \item $\gamma(|\sigma_{H^+_y,T_y}(v)|)$ and $\sigma_{H^+_y,T_y}(v)$; 
  \item $\gamma(\varphi(v))$;
  \item $\gamma(d)$;
  \item $\gamma(\varphi(u_i))$ for each $i\in\{1,\dots,d\}$;
  \item $\gamma(d_{T_{y}}(x_{T_{y}}(u_i)))$ for each $i\in\{1,\dots,d\}$;
  \item
  if $y\neq h$ then $1$, $\gamma(d_{T_{y+1}}(x_{T_{y+1}}(u_i)))$ for each $i\in\{1,\dots,d\}$;\\
  if $y=h$ then $0$;
  \item
  if $y\neq h$ then $1,\mu_y(v)$;\\
  if $y=h$ then $0$; and
  \item $a(z)$.
  \end{compactenum}
The length of the components are as follows:
(L1) is of length $\log n -\log|T_y| + \Oh(\log\log n + \log t)$,
(L2) is of length $\log|T_y| + \Oh(k+k^{-1}\log n)$,
(L3) is of length $\Oh(\log t + \log\log n)$,
(L4) is of length $\Oh(\log t)$,
(L5) is of length $\Oh(t\cdot (\log t + \log\log n))$,
(L6) and (L7) are of lengths $\Oh(t\cdot (\log\log n + \log k))$,
(L8) is of length $\Oh(k\log\log n + k\log k)$, and
(L9) is of length $\Oh(t)$.
In total, the label length is $\log n + \Oh(k\log\log n+k^{-1}\log n + k\log k + t \log\log n + t \log k + t \log t)$.
In particular, if $t$ is a fixed constant, the label length is $\log{n} + \Oh( k\log\log{n} + k^{-1}\log{n} + k\log k)$.

We remark that, using the same trick as described at the end of Section~\ref{sec:t-trees}, we can get the length of (L5) down to $\Oh(t\cdot \log\log n)$.
Assuming $k = o(\log{n})$, the total label length then becomes $\log{n} + \Oh((k+t)\log\log{n} + k^{-1}\log{n} + k\log k)$.

\subsection{Adjacency Testing}

First note that from a given label of $z=(v,y) \in V(G)$, we can decode each block of the label.
Note also that once we decode $d=|C_H(v)|$, $\sigma_{H^+_y,T_y}(v)$, $\varphi(v)$, and $\varphi(u_i)$, $d_{T_y}(x_{T_y}(u_i))$ for all $i\in\{1,\dots,d\}$, by \lemref{all-the-obs-about-bst-and-intervals}\itemref{itm:clique-on-a-single-path} we can determine $\sigma_{H^+_y,T_y}(x_{T_y}(w))$ for each $w\in\{v,u_1,\dots,u_d\}$.

Given the labels of two vertices $z_1:=(v_1,y_1)$ and $z_2:=(v_2,y_2)$ in $G$ we test if the vertices are adjacent as follows.
Looking up the value of $A(\alpha(y_1),\alpha(y_2))$, we determine which of the following cases applies:
\begin{enumerate}
  \item $|y_1-y_2|\ge 2$:
  In this case, we immediately conclude that $z_1$ and $z_2$ are not adjacent in $G$ since they are not adjacent even in $H\boxtimes P$.
  \item $y_1=y_2$:
  In this case, let $y:=y_1=y_2$.
  Note that (L2), (L3), (L4), (L5), (L7) contain all the pieces of the labels $\tau_{H^+_y,T_y}(v_1)$ and $\tau_{H^+_y,T_y}(v_2)$ from~\lemref{t-tree-labelling}.
  We compute $F(\tau_{H^+_y,T_y}(v_1),\tau_{H^+_y,T_y}(v_2))$ and determine if $v_1$ and $v_2$ are adjacent in $H^+_y$ (which is an induced subgraph of $H$).
  If $v_1$ and $v_2$ are not adjacent in $H^+_y$, then they are not adjacent in $H$ and as a consequence, $z_1=(v_1,y)$ and $z_2=(v_2,y)$ are not adjacent in $H\boxtimes P$ and thus also not adjacent in $G$.
  If $v_1$ and $v_2$ are adjacent in $H^+_y$ (and thus also in $H$), then $z_1=(v_1,y)$ and $z_2=(v_2,y)$ are adjacent in $H\boxtimes P$.
  If $F(\tau_{H^+_y,T_y}(v_1),\tau_{H^+_y,T_y}(v_2))=1$,
  we identify the position of $v_1$ on the list of vertices in $C_H(v_2)$ (by its $\varphi$-colour)
  and then finally we look up the appropriate bit in $a(z_2)$ to see whether the corresponding edge in $H\boxtimes P$ is present in $G$ or not.
  If $F(\tau_{H^+_y,T_y}(v_1),\tau_{H^+_y,T_y}(v_2))=-1$,
  we identify the position of $v_2$ on the list of vertices in $C_H(v_1)$
  and then finally we look up the appropriate bit in $a(z_1)$ to see whether the corresponding edge in $H\boxtimes P$ is present in $G$ or not.

  \item $y_1=y_2-1$:
  In this case, let $y=y_1$.
  Since $v_1 \in S_y$, by~\lemref{interval-transitions},
  we can compute $J(\sigma_{H^+_y,T_y}(v_1),\mu_y(v_1))=\sigma_{H^+_{y+1},T_{y+1}}(v_1)$.
  Now, $\sigma_{H^+_{y+1},T_{y+1}}(v_1)$ was the only missing piece of $\tau_{H^+_{y+1},T_{y+1}}(v_1)$ and
  just from the label $z_2=(v_2,y+1)$ we have $\tau_{H^+_{y+1},T_{y+1}}(v_2)$.
  We compute $F(\tau_{H^+_{y+1},T_{y+1}}(v_1),\tau_{H^+_{y+1},T_{y+1}}(v_2))$ to test whether $v_1=v_2$ or $v_1v_2\in E(H^+_{y+1})$.
  If $v_1\neq v_2$ and $v_1v_2\not\in E(H^+_{y+1})$, then $z_1$ and $z_2$ are not adjacent in $H\boxtimes P$ so they are not adjacent in $G$ either.
  If $v_1=v_2$ or $v_1v_2\in E(H^+_{y+1})$ then we know that $z_1$ and $z_2$ are adjacent in $H\boxtimes P$.
  In this case, we can now consult the relevant bit of $a(z_1)$ or $a(z_2)$ to determine if $z_1$ and $z_2$ are adjacent in $G$.

  \item $y_2=y_1-1$:  This case is symmetric to the preceding case, with the roles of $z_1$ and $z_2$ reversed.
\end{enumerate}


This completes the proof of our main result.

\begin{thm}\thmlabel{main-product}
  For every fixed $t\in\N$, the family of all graphs $G$ such that $G$ is a subgraph of $H\boxtimes P$ for some $t$-tree $H$ and some path $P$ has a $(1+o(1))\log n$-bit adjacency labelling scheme.  More precisely, for each $k\in\{5,\ldots,\left\lceil\sqrt{\log n / \log\log n}\right\rceil\}$, this graph family has an $f(n)$-bit adjacency labelling scheme with $f(n)=\log{n} + \Oh(k\log\log{n} + k^{-1}\log{n} + k\log k)$.
\end{thm}

\thmref{main} and \thmref{main-all} are immediate consequences of \thmref{main-product}, \thmref{product-structure}, and \thmref{product-structure-all}.

\section{Conclusion}
\seclabel{conclusion}

We conclude with a few remarks on the computational complexity of our labelling scheme.  Given an $n$-vertex planar graph $G$, finding an 8-tree $H$ (with mapping $f$ as in \lemref{interval-representation} and colouring $\varphi$), a path $P$, and a mapping of $G$ into a subgraph of $H\boxtimes P$ can be done in $\Oh(n\log n)$ time \cite{morin:fast}.  The process of computing the labels of $V(G)$ as described in \secref{pxp} and \secref{hxp} has a straightforward $\Oh(n\log n)$ time implementation.  Thus, the adjacency labels described in \thmref{main} are computable in $\Oh(n\log n)$ time for $n$-vertex planar graphs.

In the discussion on the adjacency test function below, we focus on the case where $t$ is a constant (which is the case in all our applications).
The adjacency testing function can then be implemented in time $\Oh(k)$  in the standard $w$-bit word RAM model, providing for binary words of length $w = \Omega(\log{n})$, bitwise logical operations, bitwise shift operations, and a most-significant-bit operation\footnote{The only purpose of the most-significant-bit operation is allow decoding of the Elias $\gamma$ code in constant time.}.
We note that \thmref{main-product} holds for a range of $k$ from $\Omega(1)$ to $\Oh(\sqrt{\log{n} / \log\log{n}})$. This yields a trade-off between adjacency test time and label length complexities. On one side, by choosing $k = \omega(1)$, we have labels of $(1+o(1))\log{n}$ bits and an adjacency test running in nearly constant time. On the other side, by selecting an adjacency test time complexity of $k = \Oh(\sqrt{\log{n}/\log\log{n}})$, the label length is minimized and has length $\log{n} + \Oh(\sqrt{\log{n}\log\log{n}})$.

%


%
%
%
%

The current result leaves two obvious directions for future work:
\begin{enumerate}
\item The precise length of the labels in \thmref{main} and \thmref{main-product} is, at best, $\log{n} + \Oh(\sqrt{\log{n}\log\log{n}})$. The only known lower bound is $\log{n} + \Omega(1)$. Closing the gap in the lower-order term remains an open problem.

\item \thmref{main-product} implies a $(1+o(1))\log{n}$-bit labelling schemes for any family of graphs that excludes an apex graph as a minor.  Can this be extended to any $K_t$-minor free family of graphs?

\item Planar graphs are a monotone family of graphs\footnote{A family $\mathcal{G}$ of graphs is \emph{monotone} if, for every $G\in \mathcal{G}$ and every (not necessarily induced) subgraph $G'\subseteq G$, $\mathcal{G}$ contains a graph isomorphic to $G'$.} and the number of $n$-vertex labelled planar graphs is $n!2^{O(n)}$ \cite{gimenez.noy:asymptotic}.  Does every monotone family of labelled graphs containing at most $n!2^{O(n)}$ $n$-vertex labelled graphs have a $(1+o(1))\log n$-bit adjacency labelling scheme? Note that this includes $K_t$-minor free graphs \cite{norine.seymour.ea:small}. This open problem is a variant of the \emph{implicit graph conjecture}, which asserts that any hereditary family of graphs that contains at most $2^{O(n\log n)}$ $n$-vertex graphs has an $O(\log n)$ bit adjacency labelling scheme \cite{kannan.naor.ea:implicit,spinrad:efficient}.
\end{enumerate}

\section*{Acknowledgement}

Part of this research was conducted during the Eighth Workshop on Geometry and Graphs, held at the Bellairs Research Institute, January~31--February~7, 2020.  We are grateful to the organizers and participants for providing a stimulating research environment.  We are particularly grateful to Tamara~Mchedlidze and David~Wood for helpful discussions. We thank the anonymous referees for their helpful comments.

\bibliographystyle{plainurl}
\bibliography{labelling}

\end{document}